\documentclass[10pt]{extarticle}
\usepackage{indentfirst}
\usepackage{url}
\usepackage{amsmath}
\usepackage{amsfonts}
\usepackage{amsthm}
\usepackage{amssymb}
\usepackage{bm}
\usepackage{graphicx}
\usepackage{subfigure}
\usepackage{color}
\usepackage{pst-plot}
\usepackage{rotating}
\usepackage{indentfirst}
\usepackage{multirow}
\makeatletter
\newif\if@restonecol
\makeatother

\usepackage[ruled,vlined,linesnumbered]{algorithm2e}

\oddsidemargin \evensidemargin
\newtheorem{remark}{\bf Remark}
\newtheorem{notation}{\bf Notation}
\newtheorem{lemma}{\bf Lemma}
\newtheorem{theorem}{\bf Theorem}
\newtheorem{example}{\bf Example}
\newtheorem{model}{\bf Model}
\newtheorem{corollary}{\bf Corollary}
\newtheorem{proposition}{\bf Proposition}
\newtheorem{definition}{\bf Definition}

\def \cV { {\mathcal V} }

\def \cM { {\mathcal M} }

\def \XH { {X^{(h)}} }
\def \Xh { X_{h} }

\newcommand{\LX}{{{\mathcal L}_X}}
\newcommand{\DD}{{{\mathcal D}_X}}
\DeclareMathOperator{\MLdegree}{MLdegree}

\DeclareMathOperator{\sMLdegree}{sMLdegree}

\title{A Probabilistic Algorithm for Computing Data-Discriminants of Likelihood Equations}

\author{Jose Israel Rodriguez and Xiaoxian Tang}

\begin{document}

\maketitle 



\begin{abstract}
An algebraic approach to the maximum likelihood estimation problem is to solve a very structured parameterized polynomial system called likelihood equations that have finitely many complex (real or non-real) solutions.  
The only solutions that are statistically meaningful are the real solutions with positive coordinates. 
In order to classify the parameters (data) according to the number of real/positive solutions, 
we study how to efficiently compute the discriminants, say data-discriminants (DD), of the likelihood equations. 
We  develop a probabilistic algorithm with three different strategies for computing DDs. 
Our implemented probabilistic algorithm based on {\tt Maple} and {\tt FGb} 
is more efficient than our previous version \cite{RT2015}
and  is also more efficient than the standard elimination for larger benchmarks.  
 By applying {\tt RAGlib} to a DD we compute, we give the real root classification of $3$ by $3$ symmetric matrix model.  
 \end{abstract}


\section{Introduction}\label{introduction}

In this work, we address the statistics problem of {\em maximum likelihood estimation} (MLE). 
A statistical model is a subset of $\mathbb{R}^{n+1}$ whose points have nonnegative coordinates summing to one; these points represent probability distributions. 
We will be interested in statistical models that are defined by
 zero sets of polynomials restricted to the positive orthant. 
 The study of such models is a central part of algebraic statistics.
 
Given a data set and statistical model, we ask, ``Which point in the model best describes the data?" 
One way to answer this question is by maximum likelihood estimation. 
By determining the point in the model that maximizes the likelihood function, one finds the probability distribution that best describes the data.

For the models we consider, the data is discrete; more specifically, the data is a list of integers representing counts. 
The likelihood function we consider on the model is given by the monomial 
$$\ell_u(p):= p_0^{u_0}p_1^{u_1}\cdots p_n^{u_n},$$
where $u=(u_0,\dots,u_n)$ is our data.

The \textit{maximum likelihood estimate} is the point in the statistical model that maximizes the likelihood function. 
To determine this estimate,  one can solve the likelihood equations to determine critical points of the likelihood function. 
This set of critical points contains all local maxima. 
So when the maximum likelihood estimate is a local maximum, it will be one  of these critical points.
When the maximum likelihood estimate is not a local maximum, then it is in the boundary (in the Euclidean topology) or the singular locus of the model. 
In those cases, one restricts to those loci's Zariski closure and performs an analysis in this restriction. 
Iterating this procedure, one solves the global optimization problem of~MLE.

The number of complex solutions to the likelihood equations (for general data) is a topological invariant called the \textit{maximum likelihood degree}.
For discrete models, the maximum likelihood degree was shown to be an Euler-characteristic in the smooth case \cite{Huh13} and recently shown to be a weighted sum of Euler characteristics in the singular case \cite{BW15}. 
In this paper, we will only consider discrete data. 
For a recent survey for the ML degree, one can refer to \cite{HS2014}.

The ML degree measures the algebraic complexity of maximum likelihood estimation by counting the number of complex solutions (real and nonreal solutions) to the likelihood equations. 
In statistics only real solutions with positive coordinates are relevant, 
and to analyze the real root classification problem, one studies a discriminant of the system. 
Discriminants are a large topic studied in algebraic geometry, and
they are important because they are used in combination with Ehresmann's Theorem (Theorem \ref{theorem2})
for real root classification. 
In principle, discriminants can be computed with Gr\"obner bases to determine
 elimination ideals \cite[Chapter 3]{CLO2007}    or with  geometric resolutions \cite{GHMMP1998}. 
In this paper we give a probabilistic interpolation algorithm to determine discriminants of the likelihood equations. 
These interpolation methods can be applied to other parameterized system of equations, for example determining the Euclidean distance degree, which was defined in \cite{DHOST13}.

Our main contribution is a probabilistic algorithm (Algorithm \ref{interpolation}) that is more efficient for larger-sized models than the standard algorithm (Algorithm \ref{dxj}). Tables \ref{table1dense}--\ref{literatureOld}  show our contribution is nontrivial with great performance increases in particular for Model 4. 
Based on the discriminants we have computed by the probabilistic algorithm, we give the real root classification for $3$ by $3$ symmetric matrix model. 
 The main differences between this paper and our ISSAC'15 paper \cite{RT2015} are listed below.

\textbullet~We have reimplemented \cite[Algorithms 1--2]{RT2015} based on {\tt Maple} and {\tt FGb}. Comparing \cite[Tables 1--2]{RT2015} and
Tables 1--2 in Section \ref{experiment}, we see the new implementation performs much better.  A typical example is Model \ref{ex:l4}. \cite{RT2015} states the old implementation of Algorithm \ref{interpolation} (Strategy 2) takes 30 days to compute the discriminant while our new 
implementation takes less than 30 minutes. 

\textbullet~In Section \ref{algorithm}, we add Section \ref{strategies} in which we discuss the motivations and effectiveness of different strategies.  
 We add Strategy 3 which improves the efficiency of sampling over Strategy 2 in some cases. 
 We also add Lemma \ref{normal} in Section \ref{lemmas}, which guarantees the correctness of Strategy 3. 
 
\textbullet~In Section \ref{experiment}, we present the new timings in  Tables \ref{table1dense}--\ref{literatureOld} based on our new implementation. We try more examples such as
Dense Models 3--6 and Models \ref{ex:l5}--\ref{ex:l9}.  We add Table \ref{literatureNew} which compares Strategy 2 and Strategy 3.

Now we motivate our study with an illustrating example.
\begin{example}[Illustrative Example]\label{ex:introduction}
 Suppose we have a weighted four-sided die such that 
the  {\em probability} $p_i$ of observing side $i$ $(i=0, 1, 2, 3)$ of the die satisfies the constraint $p_0+2p_1+3p_2-4p_3=0$. 
We toss the die 1000 times and record a $4$-dimensional {\em data} vector $(u_0, u_1, u_2, u_3)$, where $u_i$ is the number of times we observe the side $i$. 
 We want to determine the probability distribution $(p_0,p_1,p_2,p_3)\in\mathbb{R}_{>0}^4$ that best explains the data subject to the constraint by maximum likelihood estimation:
  
  \begin{center}
{\bf Maximize the {\em likelihood function}
${p_{0}^{u_{0}}p_{1}^{u_{1}}p_{2}^{u_{2}}p_{3}^{u_{3}}} $ subject to}
{\[p_0+2p_1+3p_2-4p_3=0,
p_{0}+ p_{1}+p_{2}+ p_{3}=1, \]
\[ p_{0}>0, p_{1}>0, p_{2}>0,  \text{and}\;  p_{3}>0.\]}
\end{center}

For some statistical models,  the  MLE  problem can be solved by  well known  local {\em hill climbing} algorithms such as the expectation--maximization algorithm. 
However, the local method can fail if there is more than one local maximum.  
 One way to remedy this, is by solving the system of  {\em likelihood equations} \cite{SAB2005, CHKS2006}:
  \begin{align*}
F_0&=p_{0}\lambda_1+p_{0}\lambda_2-  u_{0}& F_3&=p_{3}\lambda_1-4p_{3}\lambda_2-   u_{3}\\
F_1&=p_{1}\lambda_1+2p_{1}\lambda_2 -  u_{1}& F_4&=p_0+2p_1+3p_2-4p_3\\
F_2&=p_{2}\lambda_1+3p_{2}\lambda_2-  u_{2}& F_5&=p_{0} + p_{1} +p_{2}+p_{3} - 1
\end{align*}
where $\lambda_1$ and $\lambda_2$ are newly introduced indeterminates (Lagrange multipliers) for formulating the likelihood equations. More specifically, for given $(u_{0},  u_{1},  u_{2},  u_{3})$,  if $(p_{0},p_{1},p_{2},p_{3})$ is a critical point of the likelihood function, then there exist  numbers $\lambda_1$ and $\lambda_2$ such that $(p_{0}, p_{1},p_{2}, p_{3}, \lambda_1, \lambda_2)$ is a solution of the polynomial system.    

For general data $u=(u_{0},u_{1},u_{2},u_{3})$, the  likelihood equations have $3$ complex solutions. However, only  solutions with positive coordinates $p_{i}$
are statistically meaningful. A solution with all positive $p_i$ coordinates is said to be a positive solution.  {\em Real root classification} (RRC) and  {\em postive root classification} (PRC) are important problems:
\begin{center}
{\bf For which $u$, the polynomial system has $1, 2$ and $3$
real/positive solutions? }
\end{center}

According to the theory of computational (real) algebraic geometry \cite{BP2001, DV2005},  the number of (real/positive) solutions only changes when the
data $u_{i}$ goes across some ``special'' values (see Theorem \ref{theorem2}). The set of  ``special" $u_{i}$ is a (projective) {\em variety}  \cite[Lemma 4]{DV2005} in  ($3$ dimensional complex projective space)
$4$-dimensional complex space.

The number of real/positive solutions is uniform over each open connected component determined by the variety. In other words, the ``special'' $u_{i}$  plays the similar role as the {\em discriminant} for univariate  polynomials. 
 The first step of  RRC is calculating the ``special'' $u_i$, leading to the \textit{discriminant problem}:
\begin{center}
{\bf How to effectively compute the ``special'' $u_{i}$?} 
\end{center}

Geometrically, the ``special'' $u_{i}$ is a projection of a variety. So it can be computed by {\em elimination}.
For instance, with the command {\tt fgb\_gbasis\_elim} in {\tt FGb} \cite{fgb}, 
we  compute that the ``special" $u_{i}$ in Example \ref{linearmodel} form a hypersurface defined by a homogenous polynomial in $u_0, u_1, u_2, u_3$.
However, for many MLE problems, the elimination computation is too expensive as shown in Tables \ref{table1dense}--\ref{literatureOld} in Section \ref{experiment}. 
\end{example}%

We remark that our work can be viewed as following the numerous efforts in applying computational algebraic geometry to tackle MLE and critical points problems \cite{SAB2005, CHKS2006, BHR2007, HS2010,Uhler2012, GDP2012, FES2012, EJ2014, HS2014, HRS,Rod14}.  

The paper is organized as follows. Section \ref{definition} proves the discriminant variety of likelihood equation is projective and gives the definition of data-discriminant. Section \ref{algorithm} presents the elimination algorithm and probabilistic algorithm (with three strategies). Section \ref{experiment} shows the experimental results comparing the two algorithms and different strategies. Section \ref{rrc} discusses the future work and gives the real root classification of $3$ by $3$ symmetric matrix model. 

\section{Definitions and preliminaries}\label{definition}
In this section, we discuss how to define ``data-discriminant''.  We assume the reader is familiar with  elimination theory  \cite[Chapter 3]{CLO2007}.
 
\begin{notation}
Let ${\mathbb P}$ denote the  {\em projective closure} of the complex numbers ${\mathbb C}$.
For  homogeneous polynomials $g_1,\ldots,g_s$ in  ${\mathbb Q}[p_0, \ldots,p_n]$, 
${\mathcal V}(g_1,\ldots,g_s)$ denotes the {\em projective variety} in ${\mathbb P}^n$ defined by $g_1, \ldots, g_s$.
Let $\Delta_{n}$ denote the $n$-dimensional {\em probability simplex} 
$\{(p_0, \ldots, p_n)\in {\mathbb R}^{n+1}|p_0>0, \ldots,p_n>0, p_0+\cdots+p_n=1\}$. 
\end{notation}

\begin{definition}{\bf (Algebraic Statistical Model and Model Invariant)}\label{model}
The set $\cM$ is said to be an {\em algebraic statistical
model} if $\cM={\mathcal V}(g_1,\ldots,g_s)\cap \Delta_{n}$  where $g_1, \ldots, g_s$ 
define an irreducible generically reduced projective variety. 
Each $g_i$ is said to be a {\em model invariant} of~$\cM$. 
If the codimension of $ \cV(g_1,\ldots,g_s)$ is $k$, 
we say $\{h_1,\dots,h_k\}$ is a set of \textit{general model invariants} for $\cM$ whenever the variety $\cV(h_1,\dots,h_k)$ has $\cV(g_1,\ldots,g_s)$ as an irreducible component.
\end{definition}

For a given algebraic statistical model, there are several different ways to formulate the likelihood equations \cite{SAB2005}. In this section, we  introduce the Lagrange likelihood equations and  define the data-discriminant for this formulation. One can similarly define data-discriminants for other formulations of  likelihood~equations. 

\begin{notation}
For any $f_1, \ldots, f_m$ in the polynomial ring ${\mathbb Q}[x_1, ..., x_k]$, ${\mathcal V}_a(f_1, \ldots, f_m)$ denotes the {\em affine variety} in 
${\mathbb C}^k$ defined by $f_1, \ldots, f_m$ and 
$\langle f_1, \ldots, f_m \rangle$ denotes the  {\em ideal} generated~by $f_1, \ldots, f_m$.
For an ideal  $I$  in ${\mathbb Q}[x_1, \ldots, x_k]$,  ${\mathcal V}_a(I)$ denotes the {\em affine variety} defined by $I$. 

\end{notation}

\begin{definition}\label{LE}{\bf (Lagrange Likelihood Equations and Correspondence)}
Given an algebraic statistical model $X$ of codimension $k$ and a set of general model invariants $\{h_1,\dots,h_k\}$,
the system of polynomials below is said to be the {\em Lagrange likelihood
equations} of $X$ when set to zero:  
$$
\begin{array}{rl}
F_{0}:=&p_0(\lambda_1+\frac{\partial h_1}{\partial p_0}\lambda_2+\cdots+\frac{\partial h_k}{\partial p_0}\lambda_{k+1})-u_0\\
&\quad\vdots\\
F_{n}:=&p_n(\lambda_1+\frac{\partial h_1}{\partial p_n}\lambda_2+\cdots+\frac{\partial h_k}{\partial
p_n}\lambda_{k+1})-u_n\\
F_{n+1}:=&g_1(p_0, \ldots,p_n)\\
&\quad\vdots\\
F_{n+s} :=&g_s(p_0,\ldots,p_n)\\
F_{n+s+1}:=&p_0+\cdots+p_n-1
\end{array}$$
where 
$g_1,\ldots,g_s$ are the model invariants of $X$ and
$u_0, \ldots, u_n$, $p_0, \ldots,p_n$, $\lambda_1,\ldots,\lambda_{k+1}$ are indeterminates
(also denoted by ${\mathbf u}$, ${\mathbf p}$, ${\Lambda}$).
More specifically, $p_0,\ldots,p_n, \lambda_1,\ldots,\lambda_{k+1}$ are unknowns and $u_0,\ldots,u_n$ are parameters.

\noindent
${\mathcal V}_a(F_0, \ldots, F_{n+s+1})$, namely the set  
\[\{({\mathbf u}, {\mathbf p}, {\Lambda})\in {\mathbb C}^{n+1}\times {\mathbb C}^{n+1}\times {\mathbb C}^{k+1}|F_0=0, \ldots, F_{n+s+1}=0\},\]
 is said to be the {\em Lagrange likelihood correspondence} of $X$ and 
denoted by $\LX$. 
\end{definition}

\begin{notation}
Let $\pi$ denote the  {\em canonical projection} from the ambient space of the Lagrange likelihood correspondence  to the $\mathbb{C}^{n+1}$ associated to the ${\mathbf u}$ indeterminants $\pi$: 
 ${\mathbb C}^{n+1}\times {\mathbb
C}^{n+s+2}\rightarrow{\mathbb C}^{n+1}$.
For any set $S$ in  ${\mathbb C}^{n+1}$, ${\mathcal I}(S)$ denotes the ideal 
\[\{D\in {\mathbb Q}[{\mathbf u}]|D(a_0, \ldots, a_n)=0, \forall (a_0, \ldots, a_n)\in S\}.\]
$\overline{S}$ denotes the {\em affine closure} of $S$ in ${\mathbb C}^{n+1}$, namely ${\mathcal V}_a({\mathcal I}(S))$.

\end{notation}

Given an algebraic statistical model $\cM$ and a data vector ${\bf u}\in {\mathbb R}_{>0}^n$,  the {\em maximum likelihood estimation} problem is to
{\bf maximize the} {\em likelihood function}
{\bf $p_0^{u_0}\cdots p_n^{u_n}$ subject to $\cM$.}
The MLE problem can be solved by computing $\pi^{-1}({\mathbf u})\cap \LX$.
 More specifically, if ${\mathbf p}$ is a regular point of ${\mathcal V}(g_1,\ldots,g_s)$, then ${\mathbf p}$ is a critical point of the likelihood function if and only if there exists $\Lambda\in {\mathbb C}^{k+1}$ such that $({\mathbf u}, {\mathbf p}, {\Lambda})\in \LX$.  Theorem 1
 states that for  a general data vector ${\mathbf u}$,  $\pi^{-1}({\mathbf u})\cap \LX$
 is a finite set and the cardinality of $\pi^{-1}({\mathbf u})\cap \LX$ is constant over a dense Zariski open set, which inspires the definition of ML degree.  
 For details, see \cite{SAB2005}.

\begin{theorem} \label{MLD}\cite{SAB2005}
For an algebraic statistical model $X$, 
 there exists an affine variety $V\subset {\mathbb C}^{n+1}$ and a non-negative integer $N$  such that for any ${\mathbf u}\in {\mathbb
C}^{n+1}\backslash V$,  
\[\#\pi^{-1}({\mathbf u})\cap \LX = N.\] 
 \end{theorem}

\begin{definition}\cite{SAB2005}{\bf (ML Degree)}
For an algebraic statistical model $X$, the non-negative integer $N$ stated in Theorem \ref{MLD}
is said to be the {\em ML degree} of $X$. 
\end{definition}


\begin{definition}\label{nddv}
For an algebraic statistical model $X$ with a set of general model invariants $\{h_1,\dots,h_k\}$, suppose $F_0, \ldots, F_{n+s+1}$ are defined as  in Definition \ref{LE}. 
\noindent
Then, we have the following:

\textbullet~$\LX_{\infty}$ denotes the {\em set of non-properness} of $\pi$, i.e., the set of the $u\in \overline{\pi(\LX)}$ such that there does not exist 
a   compact neighborhood $U$ of $u$ where
$\pi^{-1}(U)\cap \LX$ is  compact;

 \textbullet~$\LX_{p}$ denotes $\overline{\pi(\LX\cap {\mathcal V}_a(\Pi_{k=0}^np_k))}$; and
 
\textbullet~$\LX_{J}$ denotes
$\overline{\pi(\LX\cap {\mathcal V}_a(J))}$ where 
$J$ denotes 
 the determinant below
 {\footnotesize\[\det \left[
\begin{matrix}
\frac{\partial F_0}{\partial p_0} & \cdots & \frac{\partial
F_0}{\partial
p_n} & \frac{\partial F_{0}}{\partial \lambda_1} & \cdots & \frac{\partial F_{0}}{\partial \lambda_{k+1}}\\
\vdots & \ddots & \vdots &\vdots & \ddots & \vdots \\
\frac{\partial F_{n}}{\partial p_0} & \cdots & \frac{\partial
F_{n}}{\partial
p_n}& \frac{\partial F_{n}}{\partial \lambda_1} & \cdots & \frac{\partial F_{n}}{\partial \lambda_{k+1}}
\\
1 & \cdots &1& 0& \cdots &0\\
\frac{\partial h_1}{\partial p_0} & \cdots & \frac{\partial
h_1}{\partial
p_n} & \frac{\partial h_{1}}{\partial \lambda_1} & \cdots & \frac{\partial h_{1}}{\partial \lambda_{k+1}}\\
\vdots & \ddots & \vdots &\vdots & \ddots & \vdots \\
\frac{\partial h_{k}}{\partial p_0} & \cdots & \frac{\partial
h_{k}}{\partial
p_n}& \frac{\partial h_{k}}{\partial \lambda_1} & \cdots & \frac{\partial h_{k}}{\partial \lambda_{k+1}}
\end{matrix}
\right]_{(n+k+2)\times (n+k+2).}
\]}
\end{definition}

The geometric meaning of $\LX_{p}$ and $\LX_{J}$ are as follows. 
The first, $\LX_{p}$, is the projection of the intersection of the Lagrange likelihood correspondence with the  coordinate hyperplanes. 
The second, $\LX_{J}$, is the projection of the intersection of the Lagrange likelihood correspondence with the hypersurface defined by $J$. Geometrically, $\LX_{J}$ is the closure of the union of the projection of the singular
locus of $\LX$ and the set of
critical values of the restriction of $\pi$ to the regular locus of $
\LX$ \cite[Definition 2]{DV2005}.

The Lagrange likelihood equations define an affine variety.  As we continuously deform the parameters $u_i$,  coordinates of a solution can tend to infinity. 
Geometrically,   $\LX_{\infty}$ is the set of the data ${\mathbf u}$ such that the Lagrange likelihood equations have some solution $({\mathbf p}, \Lambda)$ at infinity; 
this is the closure of the set of ``non-properness''  as defined in the page 1, \cite{Jelonek1999} and page 3, \cite{DS2004}.
  It is known that the set of non-properness of $\pi$ is closed and can be computed by Gr\"obner bases (see Lemma 2 and Theorem 2 in \cite{DV2005}).  

The ML degree captures the geometry of the likelihood equations  over the complex numbers. 
However, statistically meaningful solutions  occur over real numbers. 
Below, Theorem \ref{theorem2}  states that  $\LX_{\infty},$  $\LX_{J}$ and 
$\LX_{p}$ define open connected components such that the number of real/positive solutions is uniform over each open connected component. 
 Theorem 2 is a corollary of \textit{Ehresmann's theorem} for which there exists semi-algebraic statements since 1992 \cite{CS1992}.

\begin{theorem}\label{theorem2}
For an algebraic statistical model~$X$. 
If ${\mathcal C}$ is an open connected component~of 
\[{\mathbb R}^{n+1}\backslash (\LX_{\infty}\cup \LX_{J}),\]
then over  ${\mathbf u}\in
{\mathcal C}$, 
the following is constant: 
\[\#\pi^{-1}({\bf u})\cap \LX\cap {\mathbb
R}^{n+s+2}.\]
Moreover,
if ${\mathcal C}$ is an open connected component~of 
 \[{\mathbb R}^{n+1}\backslash (\LX_{\infty}\cup \LX_{J}\cup {\LX}_{p}),\] 
then over  ${\mathbf u}\in
{\mathcal C}$, 
the following is constant: 
\[\#\pi^{-1}({\mathbf u})\cap \LX\cap ({\mathbb R}_{>0}^{n+1}\times {\mathbb R}^{s+1}).\]
\end{theorem}

Before we give the definition of data-discriminant, we study the structures of  $\LX_{p}$, $\LX_J$ and $\LX_\infty$. 
Proposition \ref{sp} shows that the structure of the likelihood equations forces $\LX_p$ to be contained in the  union of coordinate hyperplanes defined by $\prod_{k=0}^n u_k$.
Proposition \ref{sj} shows that the structure of the likelihood equations forces $\LX_J\backslash \{{\bf 0}\}$ to be a projective variety. 
Similar to the proof of Proposition \ref{sj}, we can also show that the structure of the likelihood equations forces $\LX_\infty\backslash \{{\bf 0}\}$ to be a projective  variety.

\begin{proposition}\label{sp}
For any algebraic statistical model $X$, 
\[{\LX}_{p}\subset {\mathcal V}_a(\Pi_{k=0}^nu_k).\]
\end{proposition}

{\bf Proof.}
By  Definition \ref{LE},
 for any $k$ $(0\leq k\leq n)$, 
\[u_k = p_k (\lambda_1+\frac{\partial g_1}{\partial p_1}\lambda_2+\cdots+\frac{\partial g_s}{\partial p_1}\lambda_{s+1}) - F_k.\]
Hence,  
\[u_k \in \langle F_k, p_k\rangle\cap {\mathbb Q}[u_k]\subset \langle F_0, \ldots, F_{n+s+1}, p_k\rangle\cap {\mathbb C}[{\mathbf u}]\]
So 
\[{\mathcal V}_a(\langle F_0, \ldots, F_{n+s+1}, p_k\rangle\cap {\mathbb C}[{\mathbf u}])\subset {\mathcal V}_a(u_k)\]
By the Closure Theorem \cite{CLO2007}, 
\[{\mathcal V}_a(\langle F_0, \ldots, F_{n+s+1}, p_k\rangle\cap {\mathbb C}[{\mathbf u}])=\overline{\pi({\LX}\cap {\mathcal V}_a(p_k))}\]
Therefore, 
\begin{align*}
{\LX}_{p}&=
\overline{\pi({\LX}\cap {\mathcal V}_a(\Pi_{k=0}^np_k))}\\
&=\overline{\pi({\LX}\cap \cup_{k=0}^n{\mathcal V}_a(p_k))}\\
&=\cup_{k=0}^n\overline{\pi({\LX}\cap {\mathcal V}_a(p_k))}\\
&\subset \cup_{k=0}^n{\mathcal V}_a(u_k)\\
&={\mathcal V}_a(\Pi_{k=0}^nu_k). \Box
\end{align*}

\begin{remark}
Generally, 
${\LX}_{p}\neq {\mathcal V}_a(\Pi_{k=0}^nu_k)$. For example, suppose the algebraic statistical model is ${\mathcal V}_a(p_0-p_1)\cap \Delta_1$.
Then  ${\LX}_{p}=\emptyset\neq {\mathcal V}_a(u_0u_1)$.
\end{remark}

\begin{notation}
$\DD_p$ denotes the product $\Pi_{k=0}^nu_k$.
\end{notation}

\begin{proposition}\label{sj}
For an algebraic statistical model $X$, we have
${\LX}_J\backslash \{{\bf 0}\}$ is a projective variety in ${\mathbb P}^n$, where ${\bf 0}$ is the zero vector $(0, \ldots, 0)$ in ${\mathbb C}^{n+1}$.
\end{proposition}

{\bf Proof.}
By the formulation of the Lagrange likelihood equations, 
we can prove that ${\mathcal I}(\pi(\LX\cap {\mathcal V}_a(J))$ is a homogeneous ideal by the two basic facts below, which can be proved by  Definition \ref{LE} and basic algebraic geometry arguments.  

{\bf F1.} For every ${\mathbf u}$ in  $\pi(\LX\cap {\mathcal V}_a(J))$,  each  scalar multiple  
$\alpha{\mathbf u}$ is also in  $\pi(\LX\cap {\mathcal V}_a(J))$. 

{\bf F2.} For any $S\subset {\mathbb C}^{n+1}$, if for any ${\mathbf u}\in S$ and for any scalar $\alpha \in {\mathbb C}$,  $\alpha{\mathbf u}\in S$, then ${\mathcal I}(S)$ is a homogeneous ideal in~${\mathbb Q}[{\mathbf u}]$. 

That means the ideal ${\mathcal I}(\pi(\LX\cap {\mathcal V}_a(J))$ is generated by finitely many homogeneous polynomials $D_1$, $\ldots$, $D_m$. Therefore, $\LX_J={\mathcal V}_a({\mathcal I}(\pi(\LX\cap {\mathcal V}_a(J)))={\mathcal V}_a(D_1, \ldots, D_m)$. So $\LX_J\backslash \{{\bf 0}\}={\mathcal V}(D_1, \ldots, D_m)\subset {\mathbb P}^n$. $\Box$

\begin{notation}\label{ddvj}
For an algebraic statistical model $X$, 
we define the notation $\DD_J$ according to the codimension of ${\LX}_J\backslash \{\bf{0}\}$ in ${\mathbb P}^n$.

\textbullet~If the codimension is $1$, then assume
${\mathcal V}(D_1), \ldots,  {\mathcal V}(D_K)$ are the codimension $1$ irreducible components in the minimal irreducible decomposition of ${\LX}_J\backslash \{\bf{0}\}$ in ${\mathbb P}^n$ and $\langle D_1 \rangle$, $\ldots$, $\langle D_K \rangle$ are radical. 
$\DD_J$ denotes the homogeneous polynomial $\Pi_{j=1}^KD_j$.

\textbullet~If the codimension is greater than $1$, then our  convention is to take  $\DD_J =1$. 
\end{notation}

Similarly, we  use the notation $\DD_{\infty}$ to denote the projective variety $\LX_{\infty}\backslash \{\bf{0}\}$.
 Now we define the ``data-discriminant'' of Lagrange likelihood equations. 

\begin{definition}\label{dd}{\bf (Data-Discriminant)}
For a given algebraic statistics model $X$, the homogeneous polynomial $\DD_{\infty}\cdot\DD_{J}\cdot\DD_p$ is said to be the {\em data-discriminant} (DD) of Lagrange likelihood equations of $X$ and denoted by $\DD$.
\end{definition}

\begin{remark}
Note that DD can be viewed as a generalization of the ``discriminant'' for univariate  polynomials. So it is interesting to compare DD with border polynomial (BP) \cite{BP2001} and discriminant variety (DV) \cite{DV2005}.    DV and BP are defined for general parametric polynomial systems. 
DD is defined  for  the likelihood equations but can be  generalized to generic zero-dimensional  systems. 
Generally, for any square and generic zero-dimensional  system, ${\mathcal V}_a($DD$) \subset$ DV $\subset$ ${\mathcal V}_a($BP$)$.   Note that due to the special structure of likelihood equations, DD is a homogenous polynomial despite being an affine system of equations. However, generally, DV is not a projective variety and BP is not homogenous. 
\end{remark}
\begin{example}[Linear Model]\label{linearmodel}
The algebraic statistic model for the four sided die story in  Section \ref{introduction} is given by 
\[\cM={\mathcal V}(p_0+2p_1+3p_2-4p_3)\cap \Delta_{3}.\] 
The Langrange likelihood equations are the $F_0=0, \ldots, F_5=0$ shown in  Example \ref{ex:introduction}. 
The Langrange likelihood correspondence is $\LX={\mathcal V}_a(F_0, \ldots, F_5)\subset {\mathbb C}^{10}$. If we choose generic  $(u_0, u_1, u_2, u_3)\in {\mathbb C}^4$, $\pi^{-1}(u_0, u_1, u_2, u_3)\cap \LX=3$, namely the ML degree is $3$. The data-discriminant is the product of $\DD_{\infty}$, $\DD_{p}$ and $\DD_{J}$, where

$\DD_{\infty}=u_0+u_1+u_2+u_3$,
$\DD_{p}=u_0u_1u_2u_3$, and 

 $\DD_{J}=$
  {\scriptsize$441u_0^4+4998u_0^3u_1+20041u_0^2u_1^2+33320u_0u_1^3+19600u_1^4-756u_0^3u_2+
20034u_0^2u_1u_2+83370u_0u_1^2u_2+79800u_1^3u_2-5346u_0^2u_2^2+55890u_0u_1u_2^2+119025u_1^2u_2^2+4860u_0u_2^3+76950u_1u_2^3+18225u_2^4-1596u_0^3u_3-11116u_0^2u_1u_3-17808u_0u_1^2u_3+4480u_1^3u_3+7452u_0^2u_2u_3-7752u_0u_1u_2u_3+49680u_1^2u_2u_3-17172u_0u_2^2u_3+71460u_1u_2^2u_3+27540u_2^3$$u_3+2116u_0^2u_3^2+6624u_0u_1u_3^2-4224u_1^2u_3^2-9528u_0u_2u_3^2+15264u_1u_2u_3^2$\\
$+14724u_2^2u_3^2-1216u_0u_3^3-512u_1u_3^3+3264u_2u_3^3+256u_3^4$}.

By applying the well known partial cylindrical algebraic decomposition (PCAD) \cite{CH1998} method to the data-discriminant above, we get  that for any $(u_0, u_1, u_2, u_3)\in {\mathbb R}_{>0}^4$, 

\textbullet ~if $\DD_{J}(u_0, u_1, u_2, u_3)>0$, then the system of likelihood equations has $3$ distinct real solutions and $1$ of them is positive; 

\textbullet ~if $\DD_{J}(u_0, u_1, u_2, u_3)<0$, then the system of likelihood equations has exactly $1$ real solution and it is positive. 

The answer above can be verified by the {\tt RealRootClassification} \cite{BP2001, CDMMX2010} command in {\tt Maple 2015}.  In this example, the $\DD_{\infty}$ does not effect the number of real/positive solutions since it is always positive when each $u_i$ is positive. However, generally, $\DD_{\infty}$  plays an important role in real root classification. Also remark that the real root classification is equivalent to the positive root classification for this example but it is not true generally (see the example discussed in Section \ref{rrc}). 
\end{example}

\section{Algorithm}\label{algorithm}
In this section, we discuss how to compute the discriminant $\DD$. 
We assume that $X$ is the closure of a given statistical model,  $F_0, \ldots, F_{n+s+1}$ are defined as in Definition \ref{LE},
 and $J$ is defined as in  Definition~\ref{nddv}. We rename $F_0, \ldots, F_{n+s+1}$ as $F_0, \ldots, F_m$. We also rename $p_0, \ldots, p_n, \lambda_1, \ldots, \lambda_s$ as $y_0, \ldots, y_m$.
Subsection \ref{standardAlgSection} presents the standard algorithm for reference and Subsection \ref{mainResults} presents the probabilistic  algorithm. 


\subsection{Standard Algorithm}\label{standardAlgSection}
Considering the  data-discriminant as a projection drives a natural algorithm to compute it. This is the  standard elimination algorithm in symbolic computation:

\textbullet~we compute the $\LX_J$ by {\em elimination} and then get $\DD_J$ by the {\em radical equidimensional decomposition} \cite[Definition 3]{DV2005}. The algorithm is formally described in the
Algorithm \ref{dxj};

\textbullet~we compute  $\LX_\infty$ by the Algorithm {\tt PROPERNESSDEFECTS} presented in \cite{DV2005} and  then get  $\DD_\infty$ by the  radical equidimensional  decomposition.
We omit the formal description of the algorithm.

\begin{algorithm}\label{dxj}
\scriptsize 
\DontPrintSemicolon
\LinesNumbered
\SetKwInOut{Input}{input}
\SetKwInOut{Output}{output}
\Input{$F_0, \ldots,F_{m}, J$}
\Output{$\DD_J$}
${\mathcal G}_{\bf u}\leftarrow$ a set of generators of the elimination ideal $\langle F_0, \ldots F_{m}, J\rangle\cap {\mathbb Q}[{\mathbf u}]$\nllabel{elim21}\;
$\DD_J\leftarrow$  the codimension $1$ component of the equidimensional radical decomposition of $\langle {\mathcal G}_{\bf u}\rangle$\;
{\bf return} $\DD_J$
\caption{DX-J}
\end{algorithm}

Practically, 
Algorithm \ref{dxj} may not terminate
in a reasonable time before a computer reaches its memory limit.  
For instance, when using {\tt FGb} for our Gr\"obner Bases computations the memory limit is reached  
due to the large size of the intermediate computational results. 
Since the algorithms for Gr\"obner Bases in {\tt FGb}  are based on row echelon formcomputations \cite[Section 2]{FL2010},
we record the sizes of matrices generated by {\tt FGb} for Models \ref{ex:l4}--\ref{ex:l9}, see Table \ref{tableM} in Section \ref{experiment}. 
This motivates the probabilistic algorithm found in the next subsection.
 
\subsection{Probabilistic Algorithm}\label{mainResults}
In Section \ref{lemmas},  we prepare the lemmas which are used in the algorithms in Sections \ref{mainalgorithm}--\ref{strategies}.  We present the probabilistic algorithm with Strategy 1 (Algorithm \ref{interpolation}) in Section \ref{mainalgorithm}. We discuss two different strategies  (Strategy 2 and Strategy 3) for Algorithm \ref{interpolation}
in Section \ref{strategies}. 
\subsubsection{Lemmas}\label{lemmas}
We prepare the lemmas which are used in the Sections \ref{mainalgorithm} and \ref{strategies}.   
Lemma \ref{coordinate} is used to linearly transform   parameter space. 
Corollary \ref{degree1} and Lemma \ref{degree2} are used to compute the totally degree of $\DD_J$. 
Corollary \ref{csample} is used  in the sampling for interpolation. 
Lemma \ref{normal} is used to compute the degree of $\DD_J$ and to do sampling in Strategy 3. 
In the statements of Lemmas \ref{coordinate}--\ref{degree2} and Corollaries \ref{degree1}--\ref{csample}, we say 
an affine variety $V$ in ${\mathbb C}^{n}$ is {\em non-trivial} if ${\mathbb C}^n\backslash V\neq \emptyset$. 

\begin{lemma}\label{coordinate}
For any $G\in {\mathbb Q}[{\mathbf u}]$, there exists a non-trivial affine variety $V$ in ${\mathbb C}^{n}$ such that for any $(a_1, \ldots, a_n)\in {\mathbb C}^{n}\backslash V$,  the total degree of $G$ equals the degree of 
$B(t_0, t_1, \ldots, t_n)$  {\it w.r.t.} to  $t_0$, where 
\[B(t_0, t_1, \ldots, t_n) = G(t_0, a_1 t_0 + t_1, \ldots, a_n t_0 + t_n)\]
\end{lemma}

{\bf Proof.}
Suppose the total degree of $G$ is $d$ and $G_d$ is the homogeneous component of $G$ with total degree $d$.  For any  $(1, a_1, \ldots, a_n)\in {\mathbb C}^{n+1}\backslash {\mathcal V}_a(G_d)$, 
let $B(t_0, t_1, \ldots, t_n)=G(t_0, a_1 t_0 + t_1, \ldots, a_n t_0 + t_n)$. It is easily seen that 
the degree of 
$B$ {\it w.r.t.} $t_0$ equals $d$. $\Box$

\begin{corollary}\label{degree1}
For any $G\in {\mathbb Q}[{\mathbf u}]$, there exists a non-trivial affine variety $V$ in ${\mathbb C}^{2n+2}$ such that for~any \[(a_0, b_0, \ldots, a_n, b_n)\in {\mathbb C}^{2n+2}\backslash V,\]  the total degree of $G$ equals the degree of $B(t)$
where
\[B(t) = G(a_0 t+b_0, \ldots, a_n t +b_n).\]
\end{corollary}



\begin{lemma}\label{degree2}
There exists a non-trivial affine variety $V$ in ${\mathbb C}^{2n+2}$ such that for any $(a_0, b_0, \ldots, a_n, b_n)\in {\mathbb C}^{2n+2}\backslash V$,  if 
\[\langle A(t) \rangle=\langle F_0(t), \ldots, F_n(t), F_{n+1}, \ldots, F_m, J \rangle\cap {\mathbb Q}[t]\]
where  $F_i(t)$ is the polynomial by replacing $u_i$ with $a_i t+ b_i$ in $F_i$ $(i=0, \ldots, n)$
and \[B(t)=\DD_J(a_0 t+b_0, \ldots, a_n t +b_n),\]
then $\langle B(t)\rangle =\sqrt{\langle A(t)\rangle}$.
\end{lemma}

{\bf Proof.}
By the definition of $\DD_J$ (Notation \ref{ddvj}),   there exists an affine variety $V_1$ such that for any $(a_0, b_0, \ldots, a_n, b_n)\in {\mathbb C}^{2n+2}\backslash V_1$, $\langle B(t)\rangle$  is radical. Thus, we only need to show that there exists an affine variety $V_2$ in ${\mathbb C}^{2n+2}$ such that for any $(a_0, b_0, \ldots, a_n, b_n)\in {\mathbb C}^{2n+2}\backslash V_2$,  ${\mathcal V}_a(\langle B(t)\rangle) ={\mathcal V}_a(\langle A(t)\rangle)$. 

Suppose $\pi_t$ is the canonical projection: ${\mathbb C}\times {\mathbb C}^{m+1}\rightarrow {\mathbb C}$. 
For any \[t^*\in \pi_t({\mathcal V}_a(F_0(t), \ldots, F_n(t), F_{n+1}, \ldots, F_m, J)),\] 
 let $u_i^* = a_i t^*+ b_i $ (for $i=0, \ldots, n$),  then $(u_0^*, \ldots, u_n^*)\in \pi(\LX\cap {\mathcal V}_a(J))$. Hence $\DD_J(u_0^*, \ldots, u_n^*)=0$ and so $B(t^*)=0$. Thus \[B(t)\in {\mathcal I}(\pi_t({\mathcal V}_a(F_0(t), \ldots, F_n(t), F_{n+1}, \ldots, F_m, J)).\] Therefore,
\begin{align*}
{\mathcal V}_a(A(t))&={\mathcal V}_a({\mathcal I}(\pi_t({\mathcal V}_a(F_0(t), \ldots, F_n(t), F_{n+1}, \ldots, F_m, J)))\\
&\subset {\mathcal V}_a(B(t)).
\end{align*}
For any $t^*\in {\mathcal V}_a(\langle B(t)\rangle)$, let $u_i^* = a_i t^*+ b_i $ for $i=0, \ldots, n$, then 
$(u_0^*, \ldots, u_n^*)\in {\mathcal V}_a(\DD_J)\subset \LX_J$. By the Extension Theorem \cite{CLO2007},  
there exists an affine variety $V_2\subset {\mathbb C}^{2n+2}$ such that if $(a_0, b_0, \ldots, a_n, b_n)\not\in V_2$, then 
$(u_0^*, \ldots, u_n^*)\in \pi(\LX\cap \mathcal{V}_a(J))$, thus \[t^*\in \pi_t({\mathcal V}_a(F_0(t), \ldots, F_n(t), F_{n+1}, \ldots, F_m, J))\subset {\mathcal V}_a(A(t)). \Box\]

\begin{corollary}\label{csample}
There exists a non-trivial affine variety $V$ in ${\mathbb C}^{n}$ such that for any $(a_1, \ldots, a_n)\in {\mathbb C}^{n}\backslash V$,  if 
\[\langle A(u_0) \rangle=\langle F_0,F_1^* \ldots, F_n^*, F_{n+1}, \ldots, F_m, J \rangle\cap {\mathbb Q}[u_0]\]
where 
$F_i^*$ is the polynomial by replacing $u_i$ with $a_i$ in $F_i$ ($i=1, \ldots, n$)
and \[B(u_0)=\DD_J(u_0, a_1, \ldots, a_n),\]
then $\langle B(u_0)\rangle =\sqrt{\langle A(u_0)\rangle}$.
\end{corollary}

\begin{lemma}\label{normal}
Suppose $X$ is an algebraic model with ML degree $N$. 
Let $g$ generate the reduced codimension $1$ component of 
$\langle F_0, \ldots, F_{m}\rangle\cap {\mathbb Q}[u_0, \ldots, u_n, y_0]$.
If ${\tt degree}(g, y_0)=N$ and 
$D_{y_0}={\tt resultant}(g, \frac{\partial g}{\partial y_0}, y_0)$,
then $D_{y_0}\in \langle \DD_J\rangle$. 
\end{lemma}
\begin{proof}
By Notation \ref{ddvj}, 
$\langle \DD_J\rangle$ is radical. So we only need to show ${\mathcal V}(\DD_J)\subset {\mathcal V}(D_{y_0})$. 
Alternatively, we show ${\mathbb C}^{n+1}\backslash{\mathcal V}(D_{y_0})\subset {\mathbb C}^{n+1}\backslash{\mathcal V}(\DD_J)$.
For any ${\bf u}\in {\mathbb C}^{n+1}\backslash{\mathcal V}(D_{y_0})$, $g({\bf u}, y_0)$ has $N$ distinct complex solutions {\it w.r.t} $y_0$. 
By Projective Extension Theorem \cite[Page 403, Chapter 8. Corollary 10]{CLO2007}, all the $N$ distinct complex solutions can be extended in ${\mathbb P}^{m+1}$. So the likelihood equations $F_0({\bf u})=\ldots=F_{m}(\bf u)=0$ have $N$ distinct solutions in 
${\mathbb P}^{m+1}$. Therefore, ${\bf u}\not\in {\mathcal V}(\DD_J)$. 
\end{proof}

\subsubsection{Probabilistic Algorithm with Strategy 1}\label{mainalgorithm}

We first give an example for the main probabilistic algorithm.
We will discuss different ways the main probabilistic algorithm can be adjusted in Section \ref{strategies}, 
so we say the main algorithm presented in this subsection is the probabilistic algorithm with Strategy~1. 
 
\noindent
\begin{example}[Algorithm \ref{interpolation} (Strategy 1)] \label{ex:algorithm}
Suppose the radical of the elimination ideal $\langle F, J \rangle\cap {\mathbb Q}[{\bf u}]$ is generated by $D(u_0, u_1, u_2, u_3)$,  where  
$F=u_0p^3 + u_1p^2 + u_2p + u_3$ and $J=\frac{\partial F}{\partial p}=3u_0p^2 + 2u_1p + u_2$. 
With the standard elimination algorithm, we can compute $D$ to be the unique up to scaling  homogeneous polynomial 
$$27u_0^2u_3^2 - 18u_0u_1u_2u_3 + 4u_0u_2^3 + 4u_1^3u_3 - u_1^2u_2^2.$$ 
Instead, we will compute $D$ by the steps below.

First, 
we restrict our parameter space to a general $1$-dimensional affine space by replacing each parameter $u_i$ with $\square t+\square$ where $\square's$ denote a general choice. 
For example,
we substitute $u_0=7t+11, u_1=3t+2$,  $u_2=5t+6$, $u_3=4t+13$  into $F$ and $J$.
We  compute the radical of the elimination ideal $\langle F(t, p), J(t, p) \rangle\cap {\mathbb Q}[t]$  to get an ideal generated by 
$$  t^4+\frac{173086}{17315}t^3+\frac{632753}{17315}t^2+\frac{972374}{17315}t+\frac{531011}{17315}.$$ 
By  Lemma \ref{degree2} and  Corollary  \ref{degree1},
the total degree of $D$ equals the degree of the univariate polynomial above, if the $1$-dimensional affine space is general (this is because we are intersecting the hypersurface defined by $D$ with a general line). 
Similarly, we are able to compute the degrees of $D$ with respect to  $u_0, u_1, u_2, u_3$ and get $2, 3, 3, 2$, respectively.  
Since $D$ is only unique up to scaling, to interpolate, we must fix a consistent choice of scaling. 
If possible, we will fix one of the coefficients of the pure powers $u_i^{\text{deg} D}$  to be $1$. 
However, this is not possible if the coefficients of each pure power is zero, which occurs when there exists no $u_i$ such that the degree of $D$ 
 with respect to $u_i$ equals the total degree.
So when this is not possible, we will perform a linear coordinate change so that $D$ has a pure power in the new coordinates with a nonzero coefficient. 

Second, we perform a linear change of variables by the following substitution $u_i=\square v_i+\square v_1$ in $D$ to get the polynomial $D^*(v_0,v_1,v_2,v_3)$.
 For a general substitution,
 e.g.  $u_0=v_1-2v_0$, $u_1=v_1$,  $u_2=v_1-3v_2$, $u_3=v_1-5v_3$, 
  the polynomial $D^*$ has a nonzero coefficient for the pure power $v_1^{\text{deg} D}$. 
  

Now, we may assume $D^*$ has the following structure (here we have  a fixed scaling):
$$D^*=v_1^4 + C_1v_1^3 + C_2v_1^2 + C_3v_1 + C_4,$$  where $C_i$ is a homogeneous polynomial in $v_0, v_2, v_3$ with total degree $i$.
This means the monomials of $C_1$ are  $v_0, v_2, v_3$, and we have 
$C_1 = C_{11}v_0 + C_{12}v_2 + C_{13}v_3$. 
Once we determine the coefficients $C_{11}, C_{12}, C_{13}$, we are done with $C_1$. 
This will be done by interpolation; 
rather than determining $D^*$ directly, we determine $D^*$ restricted to a line which induces linear constraints on the coefficients of $D^*$. 
For example,
we substitute $v_0=13$, $v_2=4$, $v_3=5$ into $F^*(v_0, v_1, v_2, v_3, p)$ and $J^*(v_0, v_1, v_2, v_3, p)$, 
and then   
we  compute the radical of the elimination ideal 
$\langle F^*(13, v_1, 4, 5, p), J^*(13, v_1, 4, 5, p)\rangle\cap {\mathbb Q}[v_1]$
finding it is generated~by 
\begin{equation}\label{inter1}
v_1^4+\frac{243}{2}v_1^3+\frac{87939}{16}v_1^2+\frac{425385}{4}v_1+\frac{2896803}{4}.
\end{equation}
%
By  Corollary \ref{csample}, 
the polynomial in  \eqref{inter1} is $D^*(13, v_1, 4, 5)$.
Similarly, we compute to find
\begin{small}
\begin{equation}\label{inter2}
D^*(7, v_1, 3, 11)=v_1^4-168v_1^3+\frac{36873}{4}v_1^2-\frac{690201}{4}v_1+\frac{4012281}{4}
\end{equation}
\end{small}
\begin{small}
\begin{equation}\label{inter3}
D^*(2, v_1, 8, 9)=v_1^4-\frac{221}{2}v_1^3+\frac{57627}{16}v_1^2-\frac{60183}{2}v_1+68499.
\end{equation}
\end{small}%
Comparing the coefficients of $D^*$ to Equations (\ref{inter1}--\ref{inter3}), we have the linear relations below
\begin{small}
\[
13C_{11}+4C_{12}+5C_{13}=\footnotesize{-\frac{243}{2}},\quad
7C_{11}+3C_{12}+11C_{13}=\footnotesize{-168},\quad
2C_{11}+8C_{12}+9C_{13}=\footnotesize{-\frac{221}{2}}.
\]
\end{small}%
Solving the linear system yields $C_{11}=-5$, $C_{12}=\frac{3}{2}$, $C_{13}=-\frac{25}{2}$.
We similarly compute $C_2, C_3, C_4$ to determine $D^*$:
\begin{equation}\label{Dexample}
\begin{small}
\begin{array}{c}
{\frac{27}{4}v_0^2v_1^2-\frac{135}{2}v_0^2v_1v_3+\frac{675}{4}v_0^2v_3^2-5v_0v_1^3-\frac{9}{4}v_0v_1^2v_2+\frac{225}{4}v_0v_1^2v_3+}\\
{-\frac{27}{2}v_0v_1v_2^2+\frac{135}{4}v_0v_1v_2v_3}
{-\frac{675}{4}v_0v_1v_3^2+\frac{27}{2}v_0v_2^3+v_1^4+\frac{3}{2}v_1^3v_2+}\\
{-\frac{25}{2}v_1^3v_3+\frac{99}{16}v_1^2v_2^2-\frac{135}{8}v_1^2v_2
v_3+\frac{675}{16}v_1^2v_3^2-\frac{27}{4}v_1v_2^3.}
\end{array}
\end{small}%
\end{equation}
By applying the inverse linear change of coordinates $v_0=-\frac{u_0}{2}+\frac{u_1}{2}, v_1=u_1, v_2=-\frac{u_2}{3}+\frac{u_1}{3}, v_3=-\frac{u_3}{5}+\frac{u_1}{5}$ to $D^*$ and removing the denominator, 
we recover $D$. 
 \end{example}

We present the formal description of the probabilistic algorithm in Algorithm \ref{interpolation}.  We explain the main algorithm (Algorithm \ref{interpolation}) and all the sub-algorithms (Algorithms \ref{sample}--\ref{degree}) below.

{\bf Algorithm \ref{degree} (Degree).}  The probabilistic algorithm terminates correctly by  Corollary \ref{degree1}
 and Lemma~\ref{degree2}.

 {\bf Algorithm \ref{linear} (LinearOperator).}  The probabilistic algorithm terminates correctly by  Lemma \ref{coordinate}.

{\bf Algorithm \ref{sample} (Intersect).} The probabilistic algorithm terminates correctly  by  Corollary \ref{csample}.

{\bf Algorithm \ref{interpolation} (InterpolationDX-J).} 

{\bf Lines 1--5.} We compute the total degree of $\DD_J$ and the degrees of $\DD_J$ with respect to $u_0, \ldots, u_d$: $d, d_0, \ldots, d_n$   by  Algorithm \ref{degree}.  Algorithm \ref{linear} guarantees  that $d_0=d$ by applying a proper linear transformation $u_1=a_1\cdot u_0+u_1, \ldots, u_n = a_n\cdot u_0 + u_n$. 

{\bf Lines 6--7.}
 Suppose $\DD_J=u_0^{d}+C_1u_0^{d-1}+\ldots+C_{d-1}u_0+C_{d}$ where $C_1, \ldots, C_{d}\in {\mathbb Q}[u_1, \ldots, u_n]$ and the total degree of $C_j$ is $j$. 
For $j=1, \ldots, n$, we estimate all the possible monomials of $C_j$  by computing the set
\[\{u_1^{\alpha_1}\cdots u_n^{\alpha_n}|\alpha_1+\ldots+ \alpha_n = j, 0\leq \alpha_i\leq d_i\}\] 
Assume the cardinality of the set is $N_j$ and rename these monomials as ${U}_{j, 1}, \ldots, {U}_{{j, N_j}}$. Then we assume 
\[C_j = c_{j, 1}U_{j, 1}+\ldots+c_{j,N_j}U_{j, N_j}\]
where $c_{j, 1}, \ldots, c_{j, N_j}\in {\mathbb Q}$. The rest of the algorithm is to compute  $c_{j, 1}, \ldots, c_{j, N_j}$. 

{\bf Lines 8--12.} For each $j$, for $k=1,\ldots, N_j$, for a random integer vector ${\bf b}_k = (b_{k, 1}, \ldots, b_{k, n})$, 
we compute $\DD_J(u_0, {\bf b}_k)$ by Algorithm \ref{sample}.  That means to compute the function value $C_j({\bf b}_k)$ without knowing $
\DD_J$. 

 {\bf Lines 13--15.} For each $j$, we solve a square linear equation system for the unknowns $c_{j, 1}, \ldots, c_{j, N_j}$:
\begin{align*}
c_{j, 1}U_{j, 1}({{\bf b}_{k}})+\ldots+c_{j,N_j}U_{j, N_j}({\bf b}_{k})=C_j({\bf b}_k), 
\\(k=1, \ldots, N_j)
\end{align*}
It is known that we can choose nice ${\bf b}_k$ probabilistically such that the coefficient matrix of the linear equation system is non-singular. 

{\bf Lines 16.} We apply the inverse linear transformation in the parameter space to get the $\DD_J$ for the original $F_0, \ldots, F_m$.

We can also apply the interpolation idea to  Algorithm {\tt PROPERNESSDEFECTS} \cite{DV2005} and get a probabilistic algorithm to compute the $\DD_\infty$.  We omit the formal description of the algorithm.

\begin{algorithm}\label{interpolation}
\scriptsize
\DontPrintSemicolon
\LinesNumbered
\SetKwInOut{Input}{input}
\SetKwInOut{Output}{output}
\Input{$F_0, \ldots,F_{m}, J$}
\Output{ $\DD_J$}
$a_1, \ldots, a_n\leftarrow $LinearOperator$(F_0, \ldots, F_m, J)$\;
  \For {$i$ {\bf from} $1$ {\bf to} $n$}
        {
        $F'_i\leftarrow$ replace $u_i$ in $F_i$ with $a_i u_0 + u_i$}
     $NewSys\leftarrow F_0, F'_1 \ldots,F'_n, F_{n+1}, \ldots, F_m, J$\; 
    $d, d_0, \ldots, d_n\leftarrow $Degree$(NewSys)$   

 \For {$j$ {\bf from} $1$ {\bf to} $d$\nllabel{et1}}
{

Rename all the monomials of the set
\[\{u_1^{\alpha_1}\cdots u_n^{\alpha_n}|\alpha_1+\ldots+ \alpha_n = j, 0\leq \alpha_i\leq d_i\}\]
as ${U}_{j, 1}, \ldots, {U}_{{j, N_j}}$}

$N\leftarrow \max(N_1, \ldots, N_d)$\;\nllabel{i1}
\For {$k$ {\bf from} $1$ {\bf to} $N$\nllabel{et2}}
{
$b_{ k, 1}, \ldots, b_{k, n}\leftarrow$ random integers\;
$A(u_0)\leftarrow$Intersect$(NewSys, b_{k, 1}, \ldots, b_{k, n})$\;
$C^*_{d, k}, \ldots, C^*_{1, k}\leftarrow$ the coefficients of $A(u_0)$ with respect to $u_0^0, \ldots, u_0^{d-1}$\;
}
\For {$j$ {\bf from} $1$ to $d$}
{
${\mathcal M}_j\leftarrow $ $N_j\times N_j$ matrix whose $(k, r)$-entry  is  $U_{j, r}(b_{ k, 1}, \ldots, b_{k, n})$\nllabel{matrix}\;
$C_j\leftarrow ({U}_{j, 1}, \ldots, {U}_{{j, N_j}}){\mathcal M}_j^{-1}(C^*_{j, 1}, \ldots, C^*_{j, N_j})^{T}$
\nllabel{i2}}
$\DD_J\leftarrow$ replace $u_1, \ldots, u_n$ in $u_0^d + \Sigma_{i=0}^{d-1}C_{d-i}u_0^i$ with $u_1-a_1\cdot u_0, \ldots, u_n - a_n\cdot u_0$\;
{\bf Return}  $\DD_J$
\caption{({\bf Main Algorithm}) InterpolationDX-J}
\end{algorithm}

\begin{algorithm}\label{sample}
\scriptsize 
\DontPrintSemicolon
\LinesNumbered
\SetKwInOut{Input}{input}
\SetKwInOut{Output}{output}
\Input{  $F_0, \ldots,F_{m}, J$ and integers $b_1, \ldots, b_n$}
\Output{ $\DD_J(u_0, b_1, \ldots, b_n)$}
\For {$i$ {\bf from} $1$ {\bf to} $n$} {
$F_i^*\leftarrow$ replace $u_i$ in $F_k$ with $b_i$\; 
}

$A(u_0)\leftarrow$ generator of the radical of elimination ideal $\langle F_0, F_{1}^*,\ldots,F_n^*, F_{n+1}, \ldots F_{m}, J\rangle\cap {\mathbb Q}[u_0]$\nllabel{elim31}\;
{\bf return} $A(u_0)$
\caption{Intersect}
\end{algorithm}

\begin{algorithm}\label{linear}
\scriptsize
\DontPrintSemicolon
\LinesNumbered
\SetKwInOut{Input}{input}
\SetKwInOut{Output}{output}
\Input{  $F_0, \ldots,F_{m}, J$}
\Output{ $a_1, \ldots, a_n$ such that  the total degree of $\DD_J$ equals the degree of $\DD_J(u_0, a_1\cdot u_0 + u_1, \ldots, a_n\cdot u_0 + u_n)$ with respect to  $u_0$}
$d, d_0, \ldots, d_n\leftarrow $Degree$(F_0, \ldots, F_m, J)$\;
\eIf{$d=d_0$}
{{\bf return } $0, \ldots, 0$}
   {
    \For {$i$ {\bf from} $1$ {\bf to} $n$}
        { $a_i\leftarrow$ a random integer\; 
        $F'_i\leftarrow$ replace $u_i$ in $F_i$ with $a_i\cdot u_0 + u_i$}
    $NewSys\leftarrow F_0, F'_1 \ldots,F'_n, F_{n+1}, \ldots, F_m, J$\;    
    $d, d_0, \ldots, d_n\leftarrow $Degree$(NewSys)$   
     }
  {\bf return} $a_1, \ldots, a_n$
\caption{LinearOperator}
\end{algorithm}

\begin{algorithm}\label{degree}
\scriptsize
\DontPrintSemicolon
\LinesNumbered
\SetKwInOut{Input}{input}
\SetKwInOut{Output}{output}
\Input{  $F_0, \ldots,F_{m}, J$}
\Output{ $d, d_0, \ldots d_n$, where $d$ is the total degree of  $\DD_J$ and $d_i$ is the degree of $\DD_J$ with respect to each $u_i$}
\For {$i$ {\bf from} $0$ {\bf to} $n$} {

$F_0^*, \ldots, F_n^*\leftarrow$ replace $u_0, \ldots, u_{i-1}, u_{i+1}, \ldots, u_n$ in $F_0, \ldots, F_n$ with random integers\; 

$A(u_i)\leftarrow$ generator of the radical of elimination ideal $\langle F_0^*, \ldots,F_n^*, F_{n+1}, \ldots F_{m}, J\rangle\cap {\mathbb Q}[u_i]$\nllabel{elim51}\;
$d_i\leftarrow$ degree of $A(u_i)$\;
$a_i,  b_i\leftarrow$ random integers\;
}
$F_0(t), \ldots, F_n(t)\leftarrow$ replace $u_0, \ldots, u_n$ with $a_0\cdot t + b_0, \ldots, a_n\cdot t+b_n$ in $F_0, \ldots, F_n$\; 
$A(t)\leftarrow$ generator of the radical of elimination ideal $\langle F_0(t), \ldots,F_n(t), F_{n+1}, \ldots F_{m}, J\rangle\cap {\mathbb Q}[t]$\nllabel{elim52}\;
$d\leftarrow$ degree of $A(t)$\;
{\bf return} $d, d_0, \ldots, d_n$
\caption{Degree}
\end{algorithm}

\begin{remark}
According to the Notation \ref{ddvj},  when the codimension of ${\LX}_J\backslash \{\bf{0}\}$ (${\LX}_{\infty}\backslash \{\bf{0}\}$) is greater than $1$, we define $
\DD_J$ ($
\DD_{\infty}$) is $1$.
In this case, the number of real/positive solutions remains constant over a dense Zariski open set of the entire parameter space. 
\end{remark}

\subsubsection{Strategy 2 and Strategy 3}\label{strategies} 
In this section, we discuss different strategies for the probabilistic algorithm of  Section \ref{mainalgorithm}. 
Note in Algorithm \ref{interpolation},lines \ref{et1}--\ref{et2}, we compute all the possible terms in $\DD_J$ and in Algorithm \ref{interpolation}, lines \ref{i1}--\ref{i2}, we interpolate all these terms at once. 
From the timings in the ``Algorithm \ref{interpolation}-Strategy 1"  columns of Tables \ref{table1dense}-\ref{literatureOld}, we see for most of examples,  Strategy 1 does not perform better than Algorithm \ref{dxj}. 

By experiments, we find the main problem of Strategy 1 is the lifting step. In fact, it is expense to compute the inverse of ${\mathcal M}_j$ in  Algorithm \ref{interpolation}, Line \ref{matrix}, which can be a large size matrix with (large) rational entries. In order to overcome this problem, we interpolate one parameter by one parameter, say Strategy 2.  
We omit the naive formal description of the algorithm, instead
to explain how Strategy 2 works, we provide  Example \ref{ex:strategy2} . 


\begin{example}[Algorithm \ref{interpolation}  (Strategy 2)]\label{ex:strategy2}
We compute $D$ in Example \ref{ex:algorithm} by Strategy 2. 
Since the first two steps of Strategy 2 are the same as that of Strategy 1 (outlined in Example \ref{ex:algorithm}), we only show how to compute $D^*$.
Recall $D^*$ is a homogeneous polynomial in $v_0, v_1, v_2, v_3$ with total degree $4$ and the degrees of $D^*$ with respect to $v_0, v_1, v_2, v_3$ are $2, 4, 3, 2$.

First, replace all but two parameters with general numbers, e.g. $v_2=3$ and $v_3=5$, and we will compute $D^*(v_0, v_1, 3, 5)$. 
We may assume $D^*(v_0, v_1, 3, 5)$ has the following~form
$$D^*(v_0, v_1, 3, 5)=v_1^4 + C_1v_1^3+C_2v_1^2+C_3v_1+C_4,$$ where 
$C_i$ $(i=1, 2, 3, 4)$ is a univariate polynomial in $v_0$.
The degree of $C_i$ is bounded above by $i$  and the degree of $D^*$ with respect to  $v_0$. 
For instance, the degree of $C_4$ is at most $2$, so we may assume 
$$C_4=C_{40}v_0^2+C_{41}v_0+C_{42},$$
 where the coefficients $C_{40}, C_{41}, C_{42}$ are constants because $v_2$ and $v_3$ have been fixed.  
We determine these coefficients by interpolation. 
Now, we continue to have $v_2=3$,$v_3=5$, and also fix $v_0$ to be a general number, e.g. $v_0=2$, to determine the elimination ideal 
\begin{center}
$\langle F(2, v_1, 3, 5, p), J(2, v_1, 3, 5, p)\rangle\cap {\mathbb Q}[v_1]=\langle${$v_1^4-68v_1^3+\frac{5733}{4}v_1^2-\frac{36801}{4}v_1+17604$}
$\rangle$.
\end{center}
By  Corollary \ref{csample}, 
\begin{center}
$D^*(2, v_1, 3, 5)=${$v_1^4-68v_1^3+\frac{5733}{4}v_1^2-\frac{36801}{4}v_1+17604$}.
\end{center}
Similarly, by fixing $v_0$ to be  different values, we also determine
\begin{center}
$D^*(7, v_1, 3, 5)=${$v_1^4-93v_1^3+\frac{6219}{2}v_1^2-\frac{174231}{4}v_1+\frac{837081}{4}$}\\
$D^*(9, v_1, 3, 5)=${$v_1^4-103v_1^3+\frac{7749}{2}v_1^2-\frac{248103}{4}v_1+\frac{1379997}{4}$.}
\end{center}
These $3$ univariate polynomials induce linear constraints on the coefficients $C_{40},C_{41},C_{42}:$
\begin{small}
\[
2^2C_{40}+2C_{41}+C_{42}={17604},\quad
7^2C_{40}+7C_{41}+C_{42}={\frac{837081}{4}},\quad
~9^2C_{40}+9C_{41}+C_{42}={\frac{1379997}{4}}.
\]
\end{small}
%
Solving this linear system, we find 
$C_4=\frac{16875}{4}v_0^2+\frac{729}{2}v_0.$

Similarly, we compute $C_1, C_2, C_3$ and get $D^*(v_0, v_1, 3, 5)$:
\noindent
{$v_1^4+\frac{27}{4}v_0^2v_1^2-\frac{675}{2}v_0^2v_1+\frac{16875}{4}v_0^2-5v_0v_1^3+\frac{549}{2}v_0v_1^2
-{3834}v_0v_1+\frac{729}{2}v_0-58v_1^3+\frac{3429}{4}v_1^2-\frac{729}{4}v_1$.}

 Second, we compute $D^*(v_0, v_1, v_2, 5)$. According to the monomials of $D^*(v_0, v_1, 3, 5)$, we may assume $D^*(v_0,v_1,v_2,5)$ has the following form:
$${v_1^4+C_1v_0^2v_1^2+C_2v_0^2v_1+C_3v_0^2+C_4v_0v_1^3+\cdots
+C_6v_0v_1+C_7v_0+C_8v_1^3+C_9v_1^2+C_{10}v_1,}$$
\noindent
 where
$C_i$ $(i=1, \ldots, 10)$ is a polynomial in $v_2$. 
The degree of $C_i$ is bounded above by the degree of $D^*$ with respect to  $v_2$ and 
the difference 
${\tt deg}(D^*)-{\tt deg}(T_i))$ where 
$T_i$ corresponds to the term of $C_i$ in $D^*(v_0, v_1, v_2, 5)$.
For instance, the degree of $C_{10}$ is $3$, so 
$$C_{10}=C_{100}v_2^3+C_{101}v_2^2+C_{102}v_2+C_{103}.$$ Now we solve the coefficients $C_{100}, C_{101}, C_{102}, C_{103}$.
For general specializations $v_2=11, 21, 4$, we repeat the first step and compute

$D^*(v_0, v_1, 11, 5)=${$\frac{27}{4}v_0^2v_1^2-\frac{675}{2}v_0^2v_1+\frac{16875}{4}v_0^2-5v_0v_1^3+\frac{513}{2}v_0v_1^2
-{3996}v_0v_1+\frac{35937}{2}v_0+v_1^4-46v_1^3+\frac{3501}{4}v_1^2-\frac{35937}{4}v_1$. }

$D^*(v_0, v_1, 21, 5)=${$\frac{27}{4}v_0^2v_1^2-\frac{675}{2}v_0^2v_1+\frac{16875}{4}v_0^2-5v_0v_1^3+234v_0v_1^2
-\frac{13757}{2}v_0v_1-\frac{250047}{2}v_0+v_1^4-31v_1^3+\frac{4023}{2}v_1^2-\frac{250047}{4}v_1$. }

$D^*(v_0, v_1, 4, 5)=${$\frac{27}{4}v_0^2v_1^2-\frac{675}{2}v_0^2v_1+\frac{16875}{4}v_0^2-5v_0v_1^3-\frac{1089}{4}v_0v_1^2
-\frac{15039}{4}v_0v_1-864v_0+v_1^4-\frac{113}{2}v_1^3+\frac{13059}{16}v_1^2-432v_1$. }

These bivariate polynomials induce linear constraints on  $C_{100},C_{101},C_{102},C_{103}$:
\begin{center}
$3^3C_{100}+3^2C_{101}+3C_{102}+C_{103}=${\footnotesize$-\frac{729}{4}$}
\end{center}
\begin{center}
$11^3C_{100}+11^2C_{101}+11C_{102}+C_{103}=${\footnotesize$-\frac{35937}{4}$}
\end{center}
\begin{center}
$21^3C_{100}+21^2C_{101}+21C_{102}+C_{103}=${\footnotesize$-\frac{250047}{4}$}
\end{center}
\begin{center}
$4^3C_{100}+4^2C_{101}+4C_{102}+C_{103}=${\footnotesize$-432.$}
\end{center}
Solving this system determines $C_{10}$. 
Similarly, we find $C_1, \ldots, C_9$ to
get $D^*(v_0, v_1, v_2, 5)$:

\noindent
{$\frac{27}{4}v_0^2v_1^2-\frac{675}{2}v_0^2v_1+\frac{16875}{4}v_0^2-5v_0v_1^3-\frac{9}{4}v_0v_1^2v_2+\frac{1125}{4}v_0v_1^2-\frac{27}{2}v_0v_1v_2^2+\frac{675}{4}v_0v_1v_2$}
{$-\frac{16875}{4}v_0v_1+\frac{27}{2}v_0v_2^3+v_1^4+\frac{3}{2}v_1^3v_2-\frac{125}{2}v_1^3+\frac{99}{16}v_1^2v_2^2-\frac{675}{8}v_1^2v_2
+\frac{16875}{16}v_1^2-\frac{27}{4}v_1v_2^3$. }

Third, it is easy to recover $D^*$ from $D^*(v_0, v_1, v_2, 5)$ 
because we know $D^*$ is homogeneous. 
For instance, the second term in
$D^*(v_0, v_1, v_2, 5)$ is $-\frac{675}{2}v_0^2v_1$. The total degree of this term is $3$. So the corresponding term in $D^*$ is  
$\frac{-\frac{675}{2}v_0^2v_1v_3}{5}=-\frac{135}{2}v_0^2v_1v_3$. Similarly, we recover all terms of $D^*$ and again get \eqref{Dexample}.



 \end{example}
 
From the timings in the columns ``Algorithm \ref{interpolation}. Strategy 2" of Tables \ref{table1dense}--\ref{literatureOld}, we see Strategy 2 improves the efficiency for the larger examples, e.g. the  Dense Models 3--4 and Models \ref{ex:l3}--\ref{ex:l4}. The advantage of Strategy 2 compared to Strategy 1
is that we only need to solve a small-sized (bounded by the degree of $\DD_J$ with respect to each parameter at each step)
linear equation system in the lifting step. Therefore, the most time consuming part of Strategy 2 is sampling. As soon as we can compute the degree of $\DD_J$, 
we can estimate the computational timings of Strategy 2. In fact,
we can first check the approximate timing of sampling once by running Algorithm \ref{sample} once, say $T_s$. Suppose the total degree of $\DD_J$ is $d$ and the degrees with respect to  $u_0, \ldots, u_n$ are $d_0, \ldots, d_n$.  Without loss of generality, assume $d=d_0$ and $d_0\geq d_1\geq \cdots \geq d_n$. Then the total approximate sampling timing of Strategy 2 is 
\begin{equation}\label{et}
T_sd_2 \cdots d_n.
\end{equation}
 So if we improve $T_s$, we might improve the efficiency.  This motivates Strategy 3.  
  
 We give another two sub-algorithms Algorithm \ref{samples3} and Algorithm \ref{degrees3} to do sampling for interpolation and to compute the degree of $\DD_J$, respectively.  More precisely,  Strategy 3 means to use Algorithms \ref{samples3} and \ref{degrees3} instead of Algorithms \ref{sample} and \ref{degree} in Algorithm \ref{interpolation}. The correctness of Algorithms \ref{samples3} and
  \ref{degrees3} directly follows from Lemma \ref{normal}.
  We give Example \ref{ex:strategy3} to show how Algorithm \ref{degrees3} works
  (and similarly this example shows how Algorithm \ref{samples3} works). 

\begin{example}[Toy Example for Strategy 3]\label{ex:strategy3}
Here, we illustrate Algorithm \ref{degrees3} 
which is used to compute the degree of the discriminant. 
Here, we compute the total degree of $\DD_J$ to be $4$ of the linear model in Example \ref{linearmodel} by Algorithm \ref{degrees3}.  
Define $F_0, F_1, \ldots, F_5$ as those in Example \ref{ex:introduction}. 
Rename $p_0, \ldots, p_3, \lambda_1, \lambda_2$ as $y_0, \ldots, y_3, y_4, y_5$. 
The determinant of Jacobian matrix of $F_0, F_1, \ldots, F_5$ with respect to  $y_0, \ldots, y_3, y_4, y_5$, say $J$, is~
$$
\begin{array}{c}
- \lambda_1^2 p_0 p_1-4  \lambda_1^2 p_0 p_2-25  \lambda_1^2 p_0 p_3- \lambda_1^2 p_1 p_2-36  \lambda_1^2 p_1 p_3-49  \lambda_1^2 
p_2 p_3+ \lambda_1  \lambda_2 p_0 p_1
\\+8  \lambda_1  \lambda_2 p_0 p_2-125  \lambda_1  \lambda_2 p_0 p_3+3  \lambda_1  \lambda_2 p_1 p_2-144  \lambda_1  \lambda_2 p_1 p_3
-147  \lambda_1 \lambda_2p_2p_3
\\+12 \lambda_2^2p_0p_1+32 \lambda_2^2p_0p_2-150 \lambda_2^2p_0p_3+4 \lambda_2^2p_1p_2-108 \lambda_2
^2p_1p_3-98 \lambda_2^2p_2p_3.
\end{array}
$$

First, we restrict the parameter space to a general line by the substitution $u_i=\square t+\square$, e.g.
 $u_0=t+13, u_1=4t+2, u_2=9t+6, u_3=31t+5$, into $F_0, F_1, \ldots, F_5$. 
Next, we eliminate all but one unknown;
 we compute the radical of the elimination ideal $\langle F_0(t, y_0, \ldots, y_6), \ldots, F_5(t, y_0, \ldots, y_6)\rangle \cap {\mathbb Q}[t, y_0]$, say $\langle g\rangle$ where $g$ is  
$$\begin{array}{c}
20250y_0^3t^2+23400y_0^3t-22770y_0^2t^2+6760y_0^3-45961y_0^2t+\\1488
y_0t^2-18954y_0^2+20582y_0t-24t^2+16094y_0-624t-4056.
\end{array}
$$
Computing the resultant   ${\tt resultant}(g, \frac{\partial g}{\partial y_0}, y_0)$, we get $(t+13)(t+\frac{45}{26})G$,
where
$$\begin{array}{c}
G=t^4+\frac{34880663}{24297438}t^3 + \frac{1144507225}{1166277024}t^2 + \frac{391327027}{1166277024}t + \frac{244617385}{4665108096}.\end{array}
$$
By Lemma \ref{normal}, we know $\DD_J(t+13, 4t+2, 9t+6, 31t+5)$ is a factor of ${\tt resultant}(g, \frac{\partial g}{\partial y_0}, y_0)$.  For the three irreducible factors of ${\tt resultant}(g, \frac{\partial g}{\partial y_0}, y_0)$, we check by computing elimination ideal that 
\begin{center}
{$\langle F_0(t, y_0, \ldots, y_6), \ldots, F_5(t, y_0, \ldots, y_6), J, t+13\rangle \cap {\mathbb Q}[t]=\langle 1\rangle$}
\end{center}
\begin{center}
{$\langle F_0(t, y_0, \ldots, y_6), \ldots, F_5(t, y_0, \ldots, y_6), J, t+\frac{45}{26}\rangle \cap {\mathbb Q}[t]=\langle 1\rangle$}
\end{center}
\begin{center}
{$\langle F_0(t, y_0, \ldots, y_6), \ldots, F_5(t, y_0, \ldots, y_6), J, G\rangle \cap {\mathbb Q}[t]=\langle G\rangle$.}
\end{center}
Therefore we confirm $\DD_J(t+13, 4t+2, 9t+6, 31t+5)$ is $G$, which is consistent with the result computed by Algorithm \ref{degree}. Again, we conclude the total degree of $\DD_J$ is $4$.  

In  Algorithm \ref{degree} we are eliminating all unknowns but $t$ in 
$\langle F_0, \ldots, F_m, J\rangle$.
While in Algorithm \ref{degrees3} we are eliminating all unknowns but $t, y_0$ in  $\langle F_0, \ldots, F_m\rangle$ and carrying out many other eliminations to $t$ in ideals $\langle F_0, \ldots, F_m, J, G_i\rangle$. 
It is not immediately clear why Algorithm \ref{degrees3} would be an improvement over Algorithm \ref{degree}.  In fact, for the larger algebraic models, $J$ will be a huge polynomial. In this case, eliminating from 
 $\langle F_0, \ldots, F_m, J\rangle$ can be very slow. 
 Either eliminating from $\langle F_0, \ldots, F_m\rangle$  or eliminating from 
 $\langle F_0, \ldots, F_m, J, G_i\rangle$ can be faster because either removing huge polynomials or including additonal small polynomials might speed up elimination according to experience. Of course it is difficult to prove theoretically which strategy will be more efficient, the best way is to implement both strategies and run the benchmarks.  In our case, Algorithm \ref{degrees3} indeed gives improvement for some larger models, see  Models \ref{ex:l4}, \ref{ex:l6}, \ref{ex:l7}, \ref{ex:l9} in Table \ref{literatureNew}, Section \ref{experiment}.

\end{example}
\begin{algorithm}\label{samples3}
\scriptsize
\DontPrintSemicolon
\LinesNumbered
\SetKwInOut{Input}{input}
\SetKwInOut{Output}{output}
\Input{  $F_0, \ldots,F_{m}, J$ and integers $b_1, \ldots, b_n$}
\Output{ $\DD_J(u_0, b_1, \ldots, b_n)$}
\For {$i$ {\bf from} $1$ {\bf to} $n$} {
$F_i^*\leftarrow$ replace $u_i$ in $F_k$ with $b_i$\; 
}
$A_1(u_0, y_0)\leftarrow$ generator of the radical of elimination ideal $\langle F_0^*, \ldots,F_n^*, F_{n+1}, \ldots F_{m}\rangle\cap {\mathbb Q}[u_0, y_0]$\nllabel{elim61}\;
\If{${\tt degree}(A_1, y_0)=ML\;Degree$\nllabel{if5}}
{$A(u_0)\leftarrow$ {\tt resultant}$(A_1, \frac{\partial A_1}{\partial y_0}, y_0)$\;}
 $newA\leftarrow 1$\;
  \For{each irreducible factor $G(u_0)$ of $A(u_0)$}{\If{$\langle F_0^*, \ldots,F_n^*, F_{n+1}, \ldots F_{m}, J, G(u_0)\rangle\cap {\mathbb Q}[u_0]\neq \langle 1\rangle$}{$newA\leftarrow newA\cdot G(u_0)$}}
  {\bf return} $newA$\;
\caption{IntersectStrategy3}
\end{algorithm}

 \begin{algorithm}\label{degrees3}
\scriptsize
\DontPrintSemicolon
\LinesNumbered
\SetKwInOut{Input}{input}
\SetKwInOut{Output}{output}
\Input{  $F_0, \ldots,F_{m}, J$}
\Output{ $d, d_0, \ldots d_n$, where $d$ is the total degree of  $\DD_J$ and $d_i$ is the degree of $\DD_J$ {\it w.r.t} each $u_i$ $(i=0, \ldots, n)$}
\For {$i$ {\bf from} $0$ {\bf to} $n$} {

$F_0^*, \ldots, F_n^*\leftarrow$ replace $u_0, \ldots, u_{i-1}, u_{i+1}, \ldots, u_n$ in $F_0, \ldots, F_n$ with random integers\; 
$g(u_i, y_0)\leftarrow$ generator of the  radical of elimination ideal $\langle F_0^*, \ldots,F_n^*, F_{n+1}, \ldots F_{m}\rangle\cap {\mathbb Q}[u_i, y_0]$\nllabel{elim71}\;
\If{${\tt degree}(g, y_0)=ML\;Degree$\nllabel{if1}
}
{$A(u_i)\leftarrow$ {\tt resultant}$(g, \frac{\partial g}{\partial y_0}, y_0)$\;
}
  $newA\leftarrow 1$\;
  \For{each irreducible factor $G(u_i)$ of $A(u_i)$}{\If{$\langle F_0^*, \ldots,F_n^*, F_{n+1}, \ldots F_{m}, J, G(u_i)\rangle\cap {\mathbb Q}[u_i]\neq \langle 1\rangle$}{$newA\leftarrow newA\cdot G(u_i)$}}
 $d_i\leftarrow$ degree of $newA$\;
 $a_i,  b_i\leftarrow$ random integers\;
 }
$F_0(t), \ldots, F_n(t)\leftarrow$ replace $u_0, \ldots, u_n$ with $a_0\cdot t + b_0, \ldots, a_n\cdot t+b_n$ in $F_0, \ldots, F_n$\; 
$g(t, y_0)\leftarrow$ generator of the radical of elimination ideal $\langle F_0(t), \ldots,F_n(t), F_{n+1}, \ldots F_{m}\rangle\cap {\mathbb Q}[t, y_0]$\nllabel{elim72}\;
\If{${\tt degree}(g, y_0)=ML\;Degree$\nllabel{if2}} 
{
$A(t)\leftarrow$ {\tt resultant}$(g, \frac{\partial g}{\partial y_0}, y_0)$\;
}
 $newA\leftarrow 1$\;
  \For{each irreducible factor $G(t)$ of $A(t)$}{\If{$\langle F_0^*, \ldots,F_n^*, F_{n+1}, \ldots F_{m}, J, G(t)\rangle\cap {\mathbb Q}[t]\neq \langle 1\rangle$}{$newA\leftarrow newA\cdot G(t)$}}
 $d\leftarrow$ degree of $newA$\;
 {\bf return} $d, d_0, \ldots, d_n$
\caption{DegreeStrategy3}
\end{algorithm}

\begin{remark}
Strategy 3 is incomplete whenever the solutions to the likelihood equations along the general line do not have distinct $y_0$ coordinates. 
This means the ``if'' condition in Algorithm \ref{samples3}--Line \ref{if5} or in Algorithm \ref{degrees3}--Lines \ref{if1}, \ref{if2} is not satisfied. 
In principle, we can modify Strategy 3 by applying a general linear coordinate change to
$y_0, \ldots, y_m$. 
\end{remark}

\section{Implementation Details and Experimental Timings}\label{experiment}

We have implemented the probabilistic method, Algorithm \ref{interpolation}  in {\tt Maple 2015}. 
For comparisons, we have also implemented the standard Algorithm \ref{dxj} in {\tt Maple 2015}. 
We run our implementations of Algorithms \ref{dxj} and \ref{interpolation} for many examples to set benchmarks by a 3.2 GHz Inter Core i5 processor (8GB total memory) under OS X 10.9.3\footnote{See our {\tt Maple} code on the website: https://sites.google.com/site/rootclassificaiton/publications/jsc2016}. We record the timings for these examples in Tables \ref{table1dense}--\ref{literatureNew}\footnote{When a computation returned no output in $2$ hours we record ``$>$2h''}.
In this section, we give conclusions from our computational experiments and results. 




The probabilistic algorithm is implemented in three different strategies described in  Section \ref{strategies};
they are the `interpolate every term at once' strategy; `interpolate one parameter at a time' strategy; and `one variable' strategy respectively. 

 There are two kinds of benchmarks, the random models and literature models. 
 The random models, say Dense Models 1--6, are dense homogeneous polynomials generated by the {\tt Maple} command
\[{\tt randpoly}([p_0, \ldots, p_n], {\tt dense}, {\tt degree}=m),\;\;\; n=2, 3;\; m=2, 3, 4.\]
 Models \ref{ex:l1}--\ref{ex:l9} are examples presented in the literatures \cite{SAB2005, DSS2009, EJ2014}.  
See the appendix for the literature models' defining equations\footnote{The dense models are available at  https://sites.google.com/site/rootclassificaiton/publications/jsc2016}.



The most important conclusion is that our method with new implementation based on {\tt Maple 2015} and {\tt FGb} is more efficient than our implementation in  \cite{RT2015}. 
For the random models, in \cite[Table 1]{RT2015}, our previous implementation of Algorithm \ref{interpolation} was only able to compute random dense models with $n+1=3$ and degree of $X\leq 3$. 
With the new implementation, we have dramatic speed ups as shown by the second and third row of  Table \ref{table1dense}. 
Furthermore, we can compute larger dense models, see the fourth and fifth row of the Table \ref{table1dense}. 

In addition, we see in Table \ref{literatureOld} that the new implementation has improved performance for literature models as well.
Models \ref{ex:l1}--\ref{ex:l4} are exactly the same four literature models in \cite[Table 2]{RT2015}.
For instance, for Model \ref{ex:l4} \cite[Example 6]{RT2015}, the old implementation of Strategy 2 spent 30 days while the new implementation of Strategy 2 takes less than 30 minutes. 

The probabilistic algorithm performs better for the larger sized models 
(Dense Models 3--4 , Models \ref{ex:l3}--\ref{ex:l4}) 
than the standard method (the standard method performed better for the smaller models).
For Models \ref{ex:l4}--\ref{ex:l9}, while using the standard algorithm our computer reaches its memory limits after running for 2--3 hours. 
One explanation for this is seen in 
Table \ref{tableM}. We see the intermediate matrices generated by Algorithm \ref{dxj} using  {\tt FGb}  are quite large. 

For most larger literature models, the new Strategy 3 improves the sampling timings compared to Strategy 2  (see  Models \ref{ex:l4}, \ref{ex:l6}, \ref{ex:l7}, \ref{ex:l9} in Table \ref{literatureNew}). 

It is a challenging task to determine the discriminants for  Models \ref{ex:l5}--\ref{ex:l9} by the probabilistic algorithm proposed in this paper (interpolation based on linear lifting). 
One way to improve the methods is to take advantage of sparsity of the monomial support of the discriminant. 
For example, the discriminant for Model $4$ has only $1307$ terms out of the possible $6188$.

Another important take away from the tables is that the probabilistic algorithm with Strategy 1 is comparable to the standard elimination method in Table \ref{literatureOld}. We only receive great improvements when employing Strategies $2$ or $3$. 





\begin{example}As a guide to reading the tables, we summarize the first row of Table \ref{table1dense}.
The model invariant of Dense Model 1 in Table \ref{table1dense} is 
$$-37p_0^2-68p_0p_1-64p_0p_2+26p_1^2+18p_1p_2+20p_2^2.$$
We see the number of probability variables ($p_0, p_1, p_2$) is $3$ and the total degree of this model invariant is $2$. 
It is straightforward to check the ML degree
is $6$. 
By Algorithm \ref{degree} or \ref{degrees3}, we compute the total degree of $\DD_J$ is $10$. Note that when we could not compute the total degree of $\DD_J$
 we record ``{\bf ?}''. 
  \end{example}
   
 \begin{table}[h]
\tiny
\centering
{\footnotesize\caption{Computing $\DD_{J}$ for random dense models (s: seconds; h: hours).}\label{table1dense}}
\begin{tabular}{|c|c|c|c|c|c|c|c|}
\hline 
\multirow{2}{*}{Dense Models}& \multirow{2}{*}{$\# p_i$} & \multirow{2}{*}{degree  $X$} & \multirow{2}{*}{ML degree $X$} & \multirow{2}{*}{degree   $D_{X_{J}}$}  & Algorithm \ref{dxj} & \multicolumn{2}{c|}{Algorithm \ref{interpolation} (Interpolation)}\tabularnewline
 \cline{7-8} 
   &            &  & &  & (Elimination) & Strategy 1  & Strategy 2   \tabularnewline
\hline 
Dense Model 1&3& 2  & 6&10&
(4.18s) 
0.04s &
(.76s) 
0.4s  &
(.61s) 
0.4s \\ \hline
 Dense Model 2&3& 3 & 12& 24&1.2s&
(1107s) 
 2.9s  &
 (1180s) 
 2.5s \\\hline
 Dense Model 3&3 & 4&20 &44 &46.4s & 27.3s & 21.8s\\ \hline
 Dense Model 4&4&  2 &14 & 34 & 3394.6s&  \bf{$>$2h}  & 252.0s\\\hline
Dense Model 5&4&  3 & 39& 114 &\bf{$>$2h}& \bf{$>$2h}   & \bf{$>$2h}\\ \hline
Dense Model 6&4&  4 & 84& \bf{?} &\bf{$>$2h} & \bf{$>$2h} & \bf{$>$2h}  \\ \hline
\end{tabular}
\\
{Table \ref{table1dense} provides the timings of Algorithm \ref{dxj} and Algorithm \ref{interpolation} (with Strategies 1--2) for Dense Models 1--6.  
The columns ``$\# p_i$'', ``degree of $X$'', ``ML degree $X$'' and ``degree  $\DD_J$'' give 
the model information: number of probability variables, total degree of the model invariant, ML degree, and total degree of $\DD_J$ respectively. 
Timings in parenthesis are the average timings from our previous implementation. 
}
\end{table}

\begin{table}[h]
\tiny
\centering
\caption{Computing $\DD_{J}$ for literature models (s: seconds; h: hours; d: days)}\label{literatureOld}
\begin{tabular}{|c|c|c|c|c|c|c|} \hline
\multirow{2}{*}{Models}&\multirow{2}{*}{ML degree $X$}&\multirow{2}{*}{degree  $\DD_J$}& Algorithm \ref{dxj}&
\multicolumn{2}{|c|}{Algorithm \ref{interpolation} (Interpolation)}\\
\cline{5-6} 
 &&&(Elimination)&Strategy 1&Strategy 2\\ \hline
  Model \ref{ex:l1} &3&6& (11.1s)  0.23s &(5.3s) 1.4s & (6.4s) 1.7s\\\hline
  Model \ref{ex:l2} &2&4&(36446s)  4.4s &(360s) 2.4s &(56.3s)  3.3s\\\hline
  Model \ref{ex:l3} &4&14&($>$16h) 4242.3s &($>$16h)  \bf{$>$2h} &(2768s) 370.3s \\ \hline
  Model \ref{ex:l4} &6&12&($>$12d) \bf{$>$2h} &($>$30d) \bf{$>$2h} &(30d) 1501.0s \\ \hline
\end{tabular}
\\
Table \ref{literatureOld} provides the timings of Algorithm \ref{dxj} and Algorithm \ref{interpolation} (with Strategies 1--2) for Models \ref{ex:l1}--\ref{ex:l4}.
Timings in parenthesis are from our previous implementation. 
\end{table}

\begin{table}[h]
\tiny
\centering
\caption{Comparing Strategy 2  and Strategy 3 for literature models (s: seconds; d: days)}\label{literatureNew}
\begin{tabular}{|c|c|c|c|c|c|c|c|} \hline
\multirow{2}{*}{Models}&\multirow{2}{*}{ML degree $X$}&\multirow{2}{*}{degree  $\DD_J$}& \multicolumn{2}{|c|}{Sampling Timing}&\multicolumn{2}{|c|}{Total Timing}\\ 
\cline{4-7} 
&&&Strategy 2 (Algorithm \ref{sample}) & Strategy 3 (Algorithm \ref{samples3}) &Strategy 2&Strategy 3\\ \hline
 Model \ref{ex:l1} &3&6&   0.05s & 0.05s&1.7s&1.7s\\ \hline
   Model  \ref{ex:l2} &2&4&   0.04s & 0.05s&3.3s&3.9s\\ \hline
  Model \ref{ex:l3} &4&14&    0.05s & 0.06s&370.3s&438.8s\\ \hline
   Model \ref{ex:l4} &6&12&  0.8s & 0.3s&1501.0s&804.8s\\ \hline
    Model \ref{ex:l5} &12&48&   2.0s & 4.8s&$>${\it {1d}} &  $>${\it{2d}}\\ \hline
   Model  \ref{ex:l6} &10&34&   62.1s & 6.8s & $>${\it{111178d}} &$>${\it{13374d}} \\ \hline
   Model \ref{ex:l7} &10&34&    65.0s & 6.9s&$>${\it{144960d}} &  $>${\it{9203d}}\\ \hline
   Model \ref{ex:l8} &23&102&  26.8s & 53.4s&$>${\it{123d}} &$>${\it{330d}}\\ \hline
   Model \ref{ex:l9} &14&70&    294.3s & 122.2s& $>${\it{7071130d}}  &$>${\it{2637947d}}\\ \hline
\end{tabular}
\\Table \ref{literatureNew} compares Strategy 2 and Strategy 3 for Models \ref{ex:l1}--\ref{ex:l9}.  The column ``Sampling Timing'' gives the timings for doing sample once. The column ``Total Timing'' gives the timings for computing $\DD_J$.  
The italics font timing means the computation did not finish in 2 hours, but we can estimate the sampling timing providing a lower bound (Example \ref{ex:Timing}).
\end{table}

\begin{example}\label{ex:Timing}
 The  lower bound of the computational timings  are approximated by the times in italics. 
For instance,  consider Model 5 in Table \ref{literatureNew}. There are $5$ parameters $u_0, u_1, u_2, u_3, u_4$. By Algorithm \ref{degree}, the total degree of $\DD_J$ is $48$ and the degrees with respect to 
the five parameters are: $31$, $44$, $48$, $44$ and $31$.  We check the timing for doing sample one time by Algorithm \ref{sample} is $2.0s$. Then we estimate the timing for sampling of Strategy 2 is $2.0s\times 44 \times 31 \times 31 \approx 23.49h$ $(1d)$ by formula (\ref{et}) in Section \ref{strategies}. Similarly, we check  the timing for doing sample one time by Algorithm \ref{samples3} is $4.8s$ and estimate the total timing for sampling of Strategy 3 is $4.8s\times 44 \times 31 \times 31 \approx 56.37h$ $(2d)$. 
\end{example}

\begin{remark}
Whenever we compute an elimination ideal we 
use the  {\tt FGb} command {\tt fgb\_gbasis\_elim}. This is done 
in  Algorithm \ref{dxj}-Line \ref{elim21},  Algorithm \ref{sample}-Line \ref{elim31}, Algorithm \ref{degree}-Lines \ref{elim51}, \ref{elim52}, Algorithm \ref{samples3}-Line \ref{elim61}, and Algorithm \ref{degrees3}-Lines \ref{elim71}, \ref{elim72}.
\end{remark}

\begin{table}[h]
\tiny
\centering
\caption{Largest intermediate matrix size after running algorithm \ref{dxj} for 2 hours}\label{tableM}
\begin{tabular}{|c|c|c|c|c|c|c|} \hline
 Model \ref{ex:l4}& Model \ref{ex:l5}& Model \ref{ex:l6}& Model \ref{ex:l7}& Model \ref{ex:l8}& Model \ref{ex:l9}\\\hline    
 $4213938\times 5121243$  &$446567\times 635260$  &$3193378\times 5474449$  & $4314279\times 7141576$ & $1990207\times 2328577$ &$3835608\times 5561055$ \\\hline
\end{tabular}
\end{table}









\section{Future Work and Final Application}\label{rrc}
\subsection{Next steps} 

For our future work we want to determine the discriminants for larger sized problem such as Models \ref{ex:l5}--\ref{ex:l9}. 
The interpolation method discussed in this paper is basically linear lifting interpolation. Section \ref{strategies}. Formula (\ref{et}) gives the timing estimation for this linear lifting method. Although Strategy 3  improves 
$T_s$ in formula (\ref{et}), the interpolation is still expensive (see Table \ref{literatureNew}) since $d_2\cdots d_n$ is quite a large number for these models. In order to overcome this problem, one possible way is to try Newton--Hensel lifting \cite{GLS2000, Schost2003}.

We also want to develop  heuristics for checking the output of probabilistic algorithm.  For all examples we have tried in Section \ref{experiment}, the computational results of probabilistic algorithm are the same with the results computed by the standard algorithm as soon as both of  
algorithms finish the computation. For some examples such as Model \ref{ex:l4}, the standard algorithm gives no output in a reasonable time. 
However, there are ways to check the output of probabilistic algorithm. 
The first way to check is to simply run the algorithm a second time. 
The second way is to evaluate the the discriminant on a random line. This univariate polynomial should agree by doing an elimination problem that replaces the data with the random line. 
The third way to check is to find a point on the discriminant. Solving the likelihood equations for this special point, we are able to check if solutions coincide by brute force. Moreover, if the discriminant has no rational points, then one could try to use numerical homotopy continuation methods instead.

The next step in this work is the complete positive root classification. 
The discriminant is necessary, but usually not enough for classifying positive solutions (see the next subsection regarding Model \ref{ex:l4}). 
 In order to complete positive root classification, we have to consider the {\em discriminant sequence} \cite{Yang1999}.  



\subsection{Final Application}
We end the paper with the discussion of real root classification on the $3\times 3$ symmetric matrix model (Model \ref{ex:l4}).

Consider  the following story with dice. A gambler has a coin,
 and two pairs of three-sided dice. The coin and the dice are all unfair. However, the two dice in the same pair have the same weight.  He plays the same game $1000$ rounds. In each round, he first tosses the coin.  If the coin lands on side $1$,  he tosses the first pair of dice. If the coin lands on side $2$, he tosses the second pair of dice. After the $1000$ rounds,  he records a $3\times 3$ data matrix $[\overline{u}_{ij}]$ $(i, j=1, 2, 3)$ where $\overline{u}_{ij}$ is the 
 the number of times for him to get the sides $i$ and $j$ with respect to the two dice.  By the matrix $[\overline{u}_{ij}]$, he is trying to estimate the probability $\overline{p}_{ij}$ of getting the sides $i$ and $j$ with respect to the two dice. 
 
It is easy to check that the matrix 
 {\footnotesize\begin{align*}
\left[
\begin{array}{ccc}   
    \overline{p}_{11} &    \overline{p}_{12}    & \overline{p}_{13} \\   
    \overline{p}_{21} &    \overline{p}_{22}   & \overline{p}_{23}\\   
    \overline{p}_{31} &   \overline{p}_{32} &  \overline{p}_{33} 
\end{array}
\right]
\end{align*}}%
\noindent
is symmetric and has at most rank $2$.
Let 
{\footnotesize\begin{equation*}
p_{ij}=
\begin{cases}
\overline{p}_{ij} & i=j\\
\frac{1}{2}\overline{p}_{ij} & i<j
\end{cases}, \;\;
u_{ij}=
\begin{cases}
\overline{u}_{ij} & i=j\\
\overline{u}_{ij}+\overline{u}_{ji} & i<j
\end{cases}.
\end{equation*}}%
\noindent
We have an algebraic statistical model 
$\cM={\mathcal V}(g)\cap \Delta_5,$ whose  Zariski closure is $X$~where 
{\footnotesize $$\begin{array}{c}
g=\det\left[
\begin{array}{cccc}   
    2p_{11} &    p_{12}    & p_{13} \\   
    p_{12} &    2p_{22}   & p_{23}\\   
    p_{13} & p_{23} & 2p_{33} 
\end{array}
\right] \text{ and }
\Delta_{5}=\{(p_{11},\ldots,p_{33})\in {\mathbb R}_{>0}^{6}|p_{11}+ p_{12}+ p_{13}+ p_{22}+ p_{23}+ p_{33}=1\}.
\end{array}$$}

According to the Definition 2, we present the Langrange likelihood equations: 
{\begin{align*}
F_0&=p_{11}\lambda_1+(8p_{22}p_{33}-2p_{23}^2)p_{11}\lambda_2-  u_{11}\\
F_1&=p_{12}\lambda_1+(2p_{13}p_{23}-4p_{12}p_{33})p_{12}\lambda_2 -  u_{12}\\
F_2&=p_{13}\lambda_1+(2p_{12}p_{23}-4p_{13}p_{22})p_{13}\lambda_2-  u_{13}\\
F_3&=p_{22}\lambda_1+(8p_{11}p_{33}-2p_{13}^2)p_{22}\lambda_2-   u_{22}\\
F_4&=p_{23}\lambda_1+(2p_{12}p_{13}-4p_{11}p_{23})p_{23}\lambda_2- u_{23}\\
F_5&=p_{33}\lambda_1+(8p_{11}p_{22}-2p_{12}^2)p_{33}\lambda_2- u_{33}\\
F_6&=g(p_{11}, p_{12}, p_{13}, p_{22}, p_{23}, p_{33})\\
F_7&=p_{11} + p_{12} +p_{13}+p_{22}+p_{23}+p_{33}-1
\end{align*}}%
where the unknowns are $p_{11}, p_{12}, p_{13}, p_{22}, p_{23}, p_{33}, \lambda_1,\lambda_{2}$ and parameters are  
$u_{11}$, $u_{12}$, $u_{13}$, $u_{22}$, $u_{23},u_{33}$.
 
We have $8$ equations in $8$ unknowns with $6$ parameters and the ML degree is $6$  \cite{SAB2005}. 
 By the Algorithm \ref{interpolation}, we have computed $\DD_J$, which has $1307$ terms with total degree $12$. 
By a similar computation, we get $\DD_\infty$\footnote{\scriptsize See  $\DD$
 on the website: https://sites.google.com/site/rootclassificaiton/publications/DD}whose factors are  $g(u_{11}, \ldots, u_{33})$
or  positive when each $u_i$ is~positive.  

For the data-discriminant $\DD$ we have computed above, we have also computed\footnote{\scriptsize The sample points were first successfully computed by one of the 
ISSAC'15 anonymous referees.} at least one rational point (sample point) from each open connected component of $\DD\neq 0$ using  {\tt RAGlib}\cite{SS2003, HS2012, GS2014}.  
With these sample points, we can solve the real root classification problem on the open cells. 
By testing all $236$ sample points, we see that  if $g(u_{11}, \ldots, u_{33})$\\$\neq 0$, then

~\textendash~if $\DD_J(u_{11}, \ldots, u_{33})>0$, then the system has $6$ distinct real solutions and there can be $6$ positive solution or $2$ positive solutions;


~\textendash~if $\DD_J(u_{11}, \ldots, u_{33})<0$, then the system has $2$ distinct
real (positive) solutions.

With $2$ of these sample points, we see that  the sign of $\DD$ is not enough 
to classify  the positive solutions.  
For example, for the sample point  
{\footnotesize $(u_{11}=1, u_{12}=1, u_{13}=\frac{280264116870825}{295147905179352825856}, u_{22}=1, u_{23}=\frac{34089009205592922038535}{141080698675730650759168},$ $u_{33}=\frac{32898355113670387769001}{141080698675730650759168})$}, 
the system has $6$ distinct positive solutions. 
While for the sample point {\footnotesize$(u_{11} = 1, u_{12} = 1, u_{13} = 199008, u_{22} = 30, u_{23} = 2022, u_{33} =1)$}, the system has also
$6$ real solutions but only $2$ positive solutions.

\section{Acknowledgments}
{We began this project in NIMS Thematic Program 2014 on Applied Algebraic Geometry. We had the idea for this extended version in 2015 summer before SIAM Conference: Applied Algebraic Geometry when the two authors were hosted by NIMS, Daejeon, Korea.  
We  thank Professors Bernd Sturmfels,  Hoon Hong, Jonathan Hauenstein and  Frank Sottile for their valuable advice on this project. 
We thank Professors Mohab Safey EI Din and  Jean-Charles Faugere for their  software  advice on {\tt RAGlib} and  {\tt FGb} respectively. 
We thank the anonymous referees of ISSAC'15 for their insightful suggestions to greatly improve our work.
We  thank Professor Guillaume Moroz for his nice discussion and advice during ISSAC'15. We thank Professor Agnes Szanto for the reference on Shape Lemma (multivariate version). }

\bibliographystyle{abbrv}
\bibliography{DDmanuscript}

\begin{thebibliography}{10}

\bibitem{BW15}
N.~Budur and B.~Wang.
\newblock The {S}igned {E}uler {C}haracteristic of {V}ery {A}ffine {V}arieties.
\newblock {\em Int. Math. Res. Not. IMRN}, (14):5710--5714, 2015.

\bibitem{BHR2007}
M.-L.~G. Buot, S.~Ho\c{s}ten, and D.~Richards.
\newblock Counting and locating the solutions of polynomial systems of maximum
  likelihood equations, ii: The behrens-fisher problem.
\newblock {\em Statistica Sinica}, 17:1343--1354, 2007.

\bibitem{CHKS2006}
F.~Catanese, S.~Ho\c{s}ten, A.~Khetan, and B.~Sturmfels.
\newblock The maximum likelihood degree.
\newblock {\em Amer. J. Math.}, 128(3):671--697, 2006.

\bibitem{CDMMX2010}
C.~Chen, J.~H. Davemport, J.~P. May, M.~M. Maza, B.~Xia, and R.~Xiao.
\newblock Triangular decomposition of semi-algebraic systems.
\newblock In {\em Proceedings of ISSAC'10}, pages 187--194. ACM New York, 2010.

\bibitem{CH1998}
G.~E. Collins and H.~Hong.
\newblock {\em Partial Cylindrical Algebraic Decomposition for Quantifier
  Elimination}.
\newblock Springer, 1998.

\bibitem{CS1992}
M.~Coste and M.~Shiota.
\newblock Nash triviality in families of nash manifolds.
\newblock {\em Inventiones Mathematicae}, 108(1):349--368, 1992.

\bibitem{CLO2007}
D.~A. Cox, J.~Little, and D.~Oshea.
\newblock {\em Ideals, varieties, and algorithms: an introduction to
  computational algebraic geometry and commutative algebra}.
\newblock Springer, 2007.

\bibitem{DHOST13}
J.~{Draisma}, E.~{Horobet}, G.~{Ottaviani}, B.~{Sturmfels}, and R.~R. {Thomas}.
\newblock {The {E}uclidean distance degree of an algebraic variety}.
\newblock Found. Comp. Math. 2014.

\bibitem{DSS2009}
M.~Drton, B.~Sturmfels, and S.~Sullivant.
\newblock {\em Lectures on algebraic statistics}.
\newblock Springer, 2009.

\bibitem{fgb}
J.-C. Faug\`ere.
\newblock Fgb: A library for computing gr{\"o}bner bases.
\newblock In K.~Fukuda, J.~Hoeven, and M.~Joswig, editors, {\em Mathematical
  Software - ICMS 2010}, volume 6327 of {\em Lecture Notes in Computer
  Science}, pages 84--87. Springer Berlin / Heidelberg, 2010.

\bibitem{FL2010}
J.-C. Faug\`ere and S.~Lachartre.
\newblock Parallel gaussian elimination for gr\"obner bases computations in
  finite fields.
\newblock In {\em Proceedings of the 4th International Workshop on Parallel and
  Symbolic Computation}, pages 89--97. ACM, 2010.

\bibitem{FES2012}
J.-C. Faug\`ere, M.~Safey EI~Din, and P.-J. Spaenlehauer.
\newblock Gr\"obner bases and critical points: the unmixed case.
\newblock In {\em Proceedings of ISSAC'12}, pages 162--169. ACM New York, 2012.

\bibitem{GHMMP1998}
M.~Giusti, J.~Heintz, J.~E. Morais, J.~Morgenstern, and L.~M. Pardo.
\newblock Straight-line programs in geometric elimination theory.
\newblock {\em Journal of pure and applied algebra}, 124(1):101--146, 1998.

\bibitem{GLS2000}
M.~Giusti, G.~Lecerf, and B.~Salvy.
\newblock A {G}r\"obner free alternative for polynomial system solving.
\newblock {\em Journal of Complexity}, 17:154--211, 2001.

\bibitem{GS2014}
A.~Greuet and M.~Safey EI~Din.
\newblock Probabilistic algorithm for the global optimization of a polynomial
  over a real algebraic set.
\newblock {\em SIAM Journal on Optimization}, 24(3):1313--1343, 2014.

\bibitem{GDP2012}
E.~Gross, M.~Drton, and S.~Petrovi\'c.
\newblock Maximum likelihood degree of variance component models.
\newblock {\em Electronic Journal of Statistics}, 6:993--1016, 2012.

\bibitem{EJ2014}
E.~Gross and J.~I. Rodriguez.
\newblock Maximum likelihood geometry in the presence of data zeros.
\newblock In {\em Proceedings of ISSAC'14}, pages 232--239. ACM New York, 2014.

\bibitem{HRS}
J.~Hauenstein, J.~I. Rodriguez, and B.~Sturmfels.
\newblock Maximum likelihood for matrices with rank constraints.
\newblock {\em Journal of Algebraic Statistics}, 5:18--38, 2014.

\bibitem{SAB2005}
S.~Ho\c{s}ten, A.~Khetan, and B.~Sturmfels.
\newblock Solving the likelihood equations.
\newblock {\em Foundations of Computational Mathematics}, 5(4):389--407, 2005.

\bibitem{HS2010}
S.~Ho\c{s}ten and S.~Sullivant.
\newblock The algebraic complexity of maximum likelihood estimation for
  bivariate missing data.
\newblock In {\em Algebraic and geometric methods in statistics}, pages
  123--133. Cambridge University Press, 2009.

\bibitem{HS2012}
H.~Hong and M.~Safey EI~Din.
\newblock Variant quantifier elimination.
\newblock {\em Journal of Symbolic Computation}, 47(7):883--901, 2012.

\bibitem{Huh13}
J.~Huh.
\newblock The maximum likelihood degree of a very affine variety.
\newblock {\em Compos. Math.}, 149(8):1245--1266, 2013.

\bibitem{HS2014}
J.~Huh and B.~Sturmfels.
\newblock {\em Likelihood geometry}, pages 63--117.
\newblock Springer International Publishing, 2014.

\bibitem{Jelonek1999}
Z.~Jelonek.
\newblock Testing sets for properness of polynoimal mappings.
\newblock {\em Mathematische Annalen}, 315:1--35, 1999.

\bibitem{DV2005}
D.~Lazard and F.~Rouillier.
\newblock Solving parametric polynomial systems.
\newblock {\em Journal of Symbolic Computation}, 42(6):636--667, 2005.

\bibitem{Rod14}
J.~I. Rodriguez.
\newblock Maximum likelihood for dual varieties.
\newblock In {\em Proceedings of the 2014 Symposium on Symbolic-Numeric
  Computation}, SNC '14, pages 43--49, New York, NY, USA, 2014. ACM.

\bibitem{RT2015}
J.~I. Rodriguez and X.~Tang.
\newblock Data-discriminants of likelihood equations.
\newblock In {\em Proceedings of ISSAC'15}, pages 307--314. ACM, 2015.

\bibitem{SS2003}
M.~Safey EI~Din and E.~Schost.
\newblock Polar varieties and computation of one point in each connected
  component of a smooth real algebraic set.
\newblock In {\em Proceedings of ISSAC'03}, pages 224--231. ACM Press, 2003.

\bibitem{DS2004}
M.~Safey EI~Din and E.~Schost.
\newblock Properness defects of projections and computaion of in each connected
  component of a real algebraic set.
\newblock {\em Discrete and Computational Geometry}, 32(3):417--430, 2004.

\bibitem{Schost2003}
E.~Schost.
\newblock Computing parametric geometric resolutions.
\newblock {\em Applicable Algebra in Engineering, Communication and Computing},
  13(5):349--393, 2003.

\bibitem{Uhler2012}
C.~Uhler.
\newblock Geometry of maximum likelihood estimation in gaussian graphical
  models.
\newblock {\em Annals of Statistics}, 40(1):238--261, 2012.

\bibitem{Yang1999}
L.~Yang.
\newblock Recent advances on determining the number of real roots of parametric
  polynomials.
\newblock {\em Journal of Symbolic Computation}, 28(1):225--242, 1999.

\bibitem{BP2001}
L.~Yang, X.~Hou, and B.~Xia.
\newblock A complete algorithm for automated discovering of a class of
  inequality-type theorems.
\newblock {\em Science in China Series F Information Sciences}, 44(1):33--49,
  2001.

\end{thebibliography}

\footnotesize
\section*{Appendix: Literature Models}\label{appendix}

\begin{model}\cite[Random Censoring Model]{DSS2009}\label{ex:l1}
\[2p_0p_1p_2 + p_1^2p_2 + p_1p_2^2 - p_0^2p_{12} + p_1p_2p_{12}=0, \;\;\; p_0 + p_1 + p_2 + p_{12} = 1\]
\end{model}

\begin{model}\cite[$3\times 3$ Zero-Diagonal Matrix]{EJ2014}\label{ex:l2}
 {\begin{align*}
\det \left[
\begin{array}{ccc}   
    0&    p_{12}    & p_{13} \\   
    p_{21} &    0 & p_{23}\\   
    p_{31} &   p_{32} &  0
\end{array}
\right]=0,\;\;\;p_{12} + p_{13} + p_{21} + p_{23} + p_{31} + p_{32} =1
\end{align*}}
\end{model}

\begin{model}\cite[Grassmannian of $2$-planes in ${\mathbb C}^4$]{SAB2005, EJ2014}\label{ex:l3}
\[p_{12}p_{34}-p_{13}p_{24}+p_{14}p_{23}=0, \;\;\; p_{12} + p_{13} + p_{14} + p_{23} + p_{24} + p_{34} =1\]
\end{model}
\begin{model}\cite[$3\times 3$ Symmetric Matrix]{SAB2005}\label{ex:l4}
{\begin{equation*}
\det\left[
\begin{array}{cccc}   
    2p_{11} &    p_{12}    & p_{13} \\   
    p_{12} &    2p_{22}   & p_{23}\\   
    p_{13} & p_{23} & 2p_{33} 
\end{array}
\right]=0, \;\;\; p_{11} + p_{12} + p_{13} + p_{22} + p_{23} + p_{33} =1
\end{equation*}}
\end{model}

\begin{model}\cite[Bernoulli $3\times 3$ Coin]{SAB2005}\label{ex:l5}
{\footnotesize \begin{equation*}
\det\left[
\begin{array}{cccc}   
    12p_{0} &    3p_{1}    & 2p_{2} \\   
    3p_{1} &    2p_{2}   & 3p_{3}\\   
    2p_{2} & 3p_{3} & 12p_{4} 
\end{array}
\right]=0,\;\;\; p_{0} + p_{1} + p_{2} + p_{3} + p_{4}  =1
\end{equation*}}
\end{model}

\begin{model}\cite[$3\times 3$ Matrix]{SAB2005}\label{ex:l6}
{\footnotesize \begin{equation*}
\det\left[
\begin{array}{cccc}   
    p_{00} &    p_{01}    & p_{02} \\   
    p_{10} &    p_{11}   & p_{12}\\   
    p_{20} & p_{21} &  p_{22} 
\end{array}
\right]=0, \;\;\; p_{00} + p_{01} + p_{02} + p_{10} + p_{11} + p_{12} +  p_{20} + p_{21} + p_{22} =1
\end{equation*}}
\end{model}

\begin{model}\cite[Projection of $3\times 4$ Matrix]{EJ2014}\label{ex:l7}
\[p_{12}p_{23}p_{34}-p_{12}p_{24}p_{33}-p_{13}p_{22}p_{34}+p_{13}p_{24}p_{32}+p_{14}p_{22}p_{33}-p_{14}p_{23}p_{32}=0,\]
\[p_{12} + p_{13} + p_{14} + p_{22} + p_{23} + p_{24} +  p_{32} + p_{33} + p_{34} =1\]
\end{model}

\begin{model}\cite[Juke-Cantor Model, Example 18]{SAB2005}\label{ex:l8}
\[q_{000}q_{111}^2 - q_{011} q_{101} q_{110}=0, \;\;\; p_{123} + p_{dis} + p_{12} + p_{13} + p_{23}=1\]
where\\
$q_{111} = p_{123} + \frac{p_{dis}}{3} - \frac{p_{12}}{3} - \frac{p_{13}}{3} - \frac{p_{23}}{3}$, 
$q_{110} = p_{123} -  \frac{p_{dis}}{3} + p_{12} - \frac{p_{13}}{3} - \frac{p_{23}}{3}$,\\
$q_{101} = p_{123} -  \frac{p_{dis}}{3} - \frac{p_{12}}{3} + p_{13} - \frac{p_{23}}{3}$, 
$q_{011} = p_{123} -  \frac{p_{dis}}{3} - \frac{p_{12}}{3} - \frac{p_{13}}{3} + p_{23}$,\\
$q_{000} = p_{123} + p_{dis} + p_{12} + p_{13} + p_{23}$.
\end{model}

\begin{model}\cite[Example 16]{SAB2005}\label{ex:l9}
\[q_2q_7-q_1q_8=0, \;\;\;q_3q_6-q_5q_4=0, \;\;\;
p_1 + p_2 + p_3 + p_4 + p_5 + p_6 + p_7 + p_8=1\]
where\\
$q_1 = p_1 + p_2 + p_3 + p_4 + p_5 + p_6 + p_7 + p_8$, 
$q_2 = p_1 - p_2 + p_3 - p_4 + p_5 - p_6 + p_7 - p_8$,\\
$q_3 =  p_1 + p_2 - p_3 - p_4 + p_5 + p_6 - p_7 - p_8$,
$q_4 =  p_1 - p_2 - p_3 + p_4 + p_5 - p_6 - p_7 + p_8$,\\
$q_5 =  p_1 + p_2 + p_3 + p_4 - p_5 - p_6 - p_7 - p_8$,
$q_6 =  p_1 - p_2 + p_3 - p_4 - p_5 + p_6 - p_7 + p_8$,\\
$q_7 =  p_1 + p_2 - p_3 - p_4 - p_5 - p_6 + p_7 + p_8$,
$q_8 =  p_1 - p_2 - p_3 + p_4 - p_5 + p_6 + p_7 - p_8$.
\end{model}

\end{document}